\documentclass{article}
\usepackage{etoolbox}
\pdfoutput=1
\newtoggle{nips}
\newtoggle{arxiv}
\togglefalse{nips}
\toggletrue{arxiv}

\usepackage[utf8]{inputenc}
\usepackage{url}            
\usepackage{booktabs}       
\usepackage{amsfonts}       
\usepackage{nicefrac}       
\usepackage{microtype}      
\usepackage{graphicx} 
\usepackage{amsmath}
\usepackage{amsthm}
\usepackage{cleveref}
\usepackage{amssymb}
\usepackage[noend]{algorithmic}
\usepackage[ruled,vlined]{algorithm2e}
\usepackage{color-edits}
\iftoggle{arxiv}{
\usepackage{fullpage} 
}{}
\iftoggle{nips}{
	\usepackage[nonatbib]{nips_2018}	
}{}

\newtheorem{theorem}{Theorem}
\newtheorem{lemma}[theorem]{Lemma} 
\newtheorem{claim}[theorem]{Claim} 
\newtheorem{proposition}[theorem]{Proposition} 

\newtheorem{corollary}[theorem]{Corollary}
\newtheorem{definition}{Definition}

\def\<{\langle}
\def\>{\rangle}

\def\tg{\tilde{g}}
\def\tW{\tilde{W}}
\def\hx{\hat{x}}
\def\hy{\hat{y}}

\def\tW{\widetilde{W}}

\def\eps{{\varepsilon}}

\def\id{{\rm I}}
\def\sT{{\sf T}}

\def\sign{{\operatorname{\rm{sgn}}}}

\def\P{{\mathbb P}}
\def\prob{{\mathbb P}}
\def\integers{{\mathbb Z}}

\def\E{{\mathbb E}} 

\def\mbar{{\bar{m}}}
\def\etabar{{\bar{\eta}}}

\def\reals{\mathbb{R}}

\def\normal{{\sf N}}
\def\Poiss{{\sf Poiss}}

\def\Unif{{\sf Unif}}

\def\DE{{\sf DE}}
\def\mom{{\sf m}}

\def\de{{\rm d}}

\def\d{{\mathrm{d}}}

\newcommand\norm[1]{\left\lVert{#1}\right\rVert}

\def\bsl{\backslash}

\newcommand\myeqref[1]{{Eq.\,\eqref{#1}}}

\def\cX{{\cal X}}

\def\ed{\stackrel{{\rm d}}{=}}

\def\convD{{\,\stackrel{\mathrm{d}}{\Rightarrow} \,}}

\def\sG{{\sf G}}

\def\hv{\widehat{v}}

\DeclareMathAlphabet{\mathpzc}{OT1}{pzc}{m}{it}

\def\hx{\widehat{x}}

\def\hs{\widehat{s}}

\def\MMSE{{\sf MMSE}}

\def\tg{\widetilde{g}}

\def\sphere{{\mathbb{S}}}

\def\OPT{{\sf OPT}}
\def\OPT{{\sf OPT}}

\def\Unif{{\rm Uniform}}

\def\sech{{\text{sech}}}

\def\P{\mathbb{P}}
\def\E{\mathbb{E}}
\def\de{\mathrm{d}}

\title{Contextual Stochastic Block Models}

\author{Yash Deshpande\thanks{Department of Mathematics, Massachusetts Institute
of Technology} \;\and \;   Andrea Montanari\;\thanks{Departments of Electrical Engineering and Statistics, Stanford University} \and\; Elchanan Mossel\thanks{Department of Mathematics, Massachusetts Institute
of Technology} \;\and\;
Subhabrata Sen\thanks{Department of Mathematics, Massachusetts Institute of Technology}} 
\begin{document}
\maketitle
\begin{abstract}
We provide the first information theoretic tight analysis for inference of latent community structure given a sparse graph along with high dimensional node covariates, correlated with the same latent communities. Our work bridges recent theoretical breakthroughs in the detection of latent community structure without nodes covariates and a large body of empirical work using diverse heuristics for combining node covariates with graphs for inference. The tightness of our analysis implies in particular, the information theoretical necessity of combining the different sources of information. 
Our analysis holds for networks of large degrees as well as for a Gaussian version of the model. 

\end{abstract}


\section{Introduction}

Data clustering is a widely used primitive in exploratory
data analysis and summarization. These methods discover
clusters or partitions that are assumed to reflect a latent partitioning of the data
with semantic significance. In a
machine learning pipeline, results of such a clustering may then be used for downstream
supervised tasks, such as feature engineering, privacy-preserving
classification or fair allocation \cite{chaudhuri2011differentially,kumar2012radiomics,corbett2017algorithmic}.

At risk of over-simplification, there are two settings that
are popular in literature. 
In {\em graph clustering}, the dataset of $n$ objects is 
represented as a symmetric \emph{similarity matrix} $A = (A_{ij})_{1\le i, j\le n}$. For instance, $A$ can be binary, where
$A_{ij} = 1$ (or $0$) denotes that the two objects $i$, $j$ are similar 
(or not). It is, then, natural to interpret $A$ as the adjacency 
matrix of a graph. This can be carried over to non-binary settings
by considering weighted graphs. On the other hand, in more traditional (binary) {\em classification} problems, the $n$ objects are represented
as $p$-dimensional feature or covariate vectors $b_1, b_2, \cdots, b_n$. This feature
representation can be the input for a clustering method such as $k$-means, or instead used to construct a similarity matrix $A$, which 
in turn is used for clustering or partitioning.
These two representations are often taken to be mutually exclusive and, in fact, interchangeable. 
Indeed, just as feature representations can be used to 
construct similarity matrices, popular spectral methods \cite{ng2002spectral,von2007tutorial} implicitly construct
a low-dimensional feature representation from the similarity matrices. 

This paper is motivated by scenarios where the graph, or similarity,
representation $A\in\reals^{n\times n}$, and the feature representation $B = [b_1, b_2, \dots, b_n]\in \reals^{p\times n}$ provide \emph{independent}, or \emph{complementary},
information on the latent clustering of the $n$ objects. 
(Technically, we will assume that $A$ and $B$ are conditionally independent given the node labels.)
We argue that in fact in almost all practical graph clustering problems, feature representations provide complementary information of the latent clustering. This is indeed the case in many social and biological networks, see e.g.~\cite{newman2016structure} and references within.

As an example, consider the `political blogs' dataset \cite{adamic2005political}. This is a directed network of political blogs during the 2004 US presidential election, with a link between two blogs if one referred to the other. 
It is possible to just use the graph structure in order to identify political communities 
(as was done in \cite{adamic2005political}). Note however that much more data is available. 
For example we may consider an alternative feature representation of the blogs, wherein
each blog is converted to a `bag-of words' vector of its
content. This
gives a quite different, and complementary representation of 
blogs that plausibly reflects their political leaning. 
A number of approaches can be used for the simple
task of predicting leaning from the graph information
(or feature information) individually. However, 
given access to both sources, it is challenging to combine
them in a principled fashion. 

In this context, we introduce a simple statistical model
of complementary graph and high-dimensional covariate data that share latent cluster structure. This model is an intuitive
combination of two well-studied models in machine learning and
statistics: the \emph{stochastic block model}  
and the \emph{spiked covariance model} \cite{abbe2017community, holland1983stochastic,johnstone2004sparse}. 
We focus on the task of uncovering this latent structure and
make the following contributions:

\begin{description}
\item[Sharp thresholds:] 
We establish a sharp information-theoretic threshold for detecting 
the latent structure in this model. This threshold is based on 
non-rigorous, but powerful, techniques from statistical physics. 
\item[Rigorous validation:] We consider a certain `Gaussian' limit 
of the statistical model, which is of independent interest. In this 
limit, we rigorously establish the correct information-theoretic
threshold using novel Gaussian comparison inequalities. We further show convergence to the Gaussian limit predictions as the density of the graph diverges.   
\item[Algorithm:] We provide a simple, iterative 
algorithm for inference based on the belief propagation heuristic. 
For data generated from the model, we empirically demonstrate that 
the the algorithm achieves the conjectured information-theoretic 
threshold. 
\end{description}

The rest of the paper is organized as follows. 
The model and results are presented in Section
\ref{sec:results}. Further related work is discussed
in Section \ref{sec:related}.  The prediction of the threshold
from statistical physics techniques is presented
in \ref{sec:bp}, along with the algorithm.
While all proofs are 
presented in the appendix, we provide
an overview of the proofs of our rigorous results in Section 
\ref{sec:overview}.
 Finally,
we numerically validate the prediction in Section
\ref{sec:expt}. 

\section{Model and main results}\label{sec:results}

We will focus on the simple case where the $n$ objects
form two  
latent clusters of approximately
equal size, labeled
$+$ and $-$. Let $v\in\{\pm 1\}^n$ be the vector 
encoding this partitioning. Then, 
the observed data is a pair of matrices $(A^G, B)$, where $A^G$
is the adjacency matrix of the graph $G$ and $B \in \mathbb{R}^{p \times n} $ is the matrix
of covariate information. Each column $b_i$, $i\le n$ of matrix $B$ contains the covariate information about vertex $i$.  We use the following probabilistic
model: conditional on $v$, and a latent vector $u\sim\normal(0, I_p/p)$:
\begin{align}
\prob(A^G_{ij}=1)  &= 
\begin{cases}
c_{{\rm in}}/n &\text{ with probability } \, ,\\
c_{{\rm out}}/n &\text{ otherwise. } 
\end{cases} \label{eq:graphmodel} \\
b_i  &= \sqrt\frac{\mu}{n} v_i u + \frac{Z_i}{\sqrt{p}}, \label{eq:covariatemodel}
\end{align}
where $Z_i\in \reals^{p}$ has independent standard normal entries.
It is convenient to parametrize the edge probabilities by the average degree $d$ and the normalized degree separation $\lambda$:
\begin{align}
c_{{\rm in}} = d + \lambda \sqrt{d}\, ,\;\;\;\;\;
c_{{\rm out}} = d - \lambda \sqrt{d}\, .
\end{align}
Here  $d$, $\lambda$, $\mu$ are parameters
of the model which, for the sake of simplicity, we assume to be fixed
and known. In other words, two objects $i, j$ in the same
cluster or community are \emph{slightly more likely} 
to be connected than for objects $i, j'$ in
different clusters. Similarly, according to $\eqref{eq:covariatemodel}$, they have 
\emph{slightly positively} correlated feature vectors $b_i$, $b_j$,
while objects $i, j'$ in different clusters have negatively
correlated covariates $b_i, b_{j'}$. 

Note that this model is a combination of two
observation
models that have been extensively studied: the stochastic
block model and the spiked covariance model. The stochastic block
model has its roots in sociology literature \cite{holland1983stochastic}
and has witnessed a resurgence of interest from the computer science
and statistics community since the work of 
Decelle et al. \cite{decelle2011asymptotic}. This work focused
on the sparse setting where the graph as $O(n)$ edges and
conjectured,
using the non-rigorous cavity method, 
the following \emph{phase transition} phenomenon. This
was 
later established rigorously in a series of papers \cite{mossel2015reconstruction,mossel2013proof,massoulie2014community}.
\begin{theorem}[\cite{mossel2015reconstruction,mossel2013proof,massoulie2014community}]\label{thm:sbm}
Suppose $d >1$ is fixed. The graph $G$ is distinguishable with high probability from an Erd\"os-Renyi random
graph with average degree $d$ if and only if $\lambda \ge  1$. Moreover, if $\lambda > 1$,
there exists a polynomial-time computable estimate $\hv = \hv(A^G)\in\{\pm 1\}^n$  of the cluster assignment satisfying, almost surely:
\begin{align}
\liminf_{n\to \infty} \frac{|\<\hv, v\>|}{n} \ge \eps(\lambda) > 0.
\end{align}
\end{theorem}
In other words, given the graph $G$, it is possible to non-trivially
estimate the latent clustering $v$ if, and only if, $\lambda > 1$. 

The covariate model \eqref{eq:covariatemodel} was
proposed by Johnstone and Lu \cite{johnstone2004sparse} and has been extensively studied in statistics and random matrix 
theory. The weak recovery threshold was characterized by a number of authors, including Baik et al \cite{baik2005phase}, Paul \cite{paul2007asymptotics} and Onatski et al \cite{onatski2013asymptotic}.
\begin{theorem}[\cite{baik2005phase,paul2007asymptotics,onatski2013asymptotic}] \label{thm:baik}
Let $\hv_1$ be the principal eigenvector
of $B^\sT B$, where $\hv_1$ is normalized so that $\norm{\hv_1}^2 = n$.
Suppose that $p, n\to\infty$  with $p/n\to 1/\gamma \in (0, \infty)$. Then
$ \liminf_{n\to\infty} |\<\hv_1, v\>|/n > 0$ if and only if $\mu > \sqrt{\gamma}$.
Moreover, if $\mu < \sqrt{\gamma}$, no such estimator exists. 
\end{theorem}
In other words, this theorem shows that it is possible to estimate $v$ nontrivally solely
from the covariates using, in fact, a spectral method if, and only if
$\mu > \sqrt{\gamma}$. 

Our first result is the following prediction that establishes
the analogous threshold prediction that smoothly interpolates
between Theorems \ref{thm:sbm} and \ref{thm:baik}. 

\begin{claim}[Cavity prediction]  \label{claim:cavity} Given $A^G, B$ as in
Eqs.\eqref{eq:graphmodel}, \eqref{eq:covariatemodel}, and
assume that $n, p\to \infty$ with $p/n \to 1/\gamma \in (0, \infty)$. Then there
exists an estimator $\hv = \hv(A^G, B) \in \{\pm 1\}^n$ 
so that $\liminf |\<\hv, v\>|/n$ is bounded away 
from $0$ \emph{ if and only if} 
\begin{align}
\lambda^2 +  \frac{\mu^2}{\gamma} > 1\, . 
\end{align}
\end{claim}
We obtain this prediction via the cavity method, a 
powerful technique from the statistical physics of mean
field models \cite{MezardMontanari}. This derivation is 
outlined in Section \ref{sec:bp}. Theorems \ref{thm:sbm} and \ref{thm:baik} confirm this prediction rigorously in the corner cases, in which either $\lambda$ or $\mu$ vanishes, using sophisticated tools from random matrix theory and sparse random graphs. 

Our main result confirms rigorously this claim in the limit of large degrees.
\begin{theorem}\label{thm:graph}
Suppose $v$ is uniformly distributed in $\{\pm1\}^n$ and we observe
$A^G, B$ as in \eqref{eq:graphmodel}, \eqref{eq:covariatemodel}. 
Consider the limit $p, n \to \infty$ with $p/n\to 1/\gamma$. Then we have, for
some $\eps(\lambda,\mu)>0$ independent of $d$, 
\begin{align}
\liminf_{n\to\infty} \sup_{\hv(\,\cdot\,)}\frac{|\<\hv(A^G,B),v\>|}{n} & \ge 
\eps(\lambda,\mu)-o_d(1) & \mbox{if $\lambda^2+\mu^2/\gamma>1$,}\\
\limsup_{n\to\infty} \sup_{\hv(\,\cdot\,)}\frac{|\<\hv(A^G,B),v\>|}{n} & =
o_d(1)& \mbox{if $\lambda^2+\mu^2/\gamma<1$.}
\end{align}
Here the limits hold in probability, the supremum is over estimators $\hv: (A^G,B)\mapsto \hv(A^G,B)\in\reals^n$, with $\|\hv(A^G,B)\|_2 = \sqrt{n}$. Here $o_d(1)$ indicates a term independent of $n$ which tends to zero as $d \to \infty$. 
\end{theorem}

In order to establish this result, we consider a modification
of the original model in \eqref{eq:graphmodel}, \eqref{eq:covariatemodel}, which is of independent interest. Suppose, 
conditional on $v\in \{\pm 1\}$ and the latent vector $u$
we observe $(A, B)$ as follows:
\begin{align}
A_{ij} & \sim \begin{cases}
\normal(\lambda v_i v_j/n, 1/n) &\text{ if } i < j\\
\normal(\lambda v_i v_j/n, 2/n )&\text{ if } i = j,
\end{cases} \label{eq:gaussiangraphmodel}\\
B_{ai} & \sim \normal( \sqrt{\mu} v_i u_a/\sqrt{n}, 1/p) \label{eq:gaussiancovariatemodel}.
\end{align}
This model differs from \eqref{eq:graphmodel}, in that the graph
observation $A^G$ is replaced by the observation $A$ which
is equal to $\lambda v v^{\sT}/n$, corrupted by Gaussian noise. 
This model generalizes so called `rank-one deformations' of random matrices
\cite{peche2006largest,knowles2013isotropic,benaych2011eigenvalues}, as well as the $\integers_2$ synchronization
model \cite{abbe2014decoding,cucuringu2015synchronization}. 

Our main motivation for introducing the Gaussian observation model is that 
it captures the large-degree  behavior of the original graph model.
The next result formalizes this intuition: its proof is an immediate generalization 
of the Lindeberg interpolation method of \cite{deshpande2016asymptotic}.
\begin{theorem}\label{thm:Universality}
Suppose $v\in \{\pm 1\}^n$ is uniformly random, and $u$ is independent. 
We denote by $I(v ; A^G, B)$ the \emph{mutual information} of
the latent random variables $v$ and the observable data $A^G, B$.
For all $\lambda, \mu$: 
 we have that:
\begin{align}
\lim_{d\to\infty}\limsup_{n\to\infty} \frac{1}{n} |I(v ; A^G, B) - I(v;  A, B)| &= 0, \\
\lim_{d\to\infty} \limsup_{n\to \infty} 
\Big|
\frac{1 }{n} \frac{\d I(v; A^G, B)}{\d (\lambda^2)} -
\frac{1}{4}\MMSE(v; A^G, B)\Big| & = 0,
\end{align}
where $\MMSE(v; A^G, B) = n^{-2}\E\{\lVert vv^\sT - \E\{vv^\sT | A^G, B\} \rVert_F^2\}$.
\end{theorem}

For the Gaussian observation model \eqref{eq:gaussiangraphmodel}, \eqref{eq:gaussiancovariatemodel} we can establish a precise weak recovery threshold,
which is the main technical novelty of this paper.
\begin{theorem}\label{thm:gaussian}
Suppose $v$ is uniformly distributed in $\{\pm1\}^n$ and we observe
$A, B$ as in \eqref{eq:gaussiangraphmodel}, \eqref{eq:gaussiancovariatemodel}. Consider the limit $p, n \to \infty$ with $p/n\to 1/\gamma$.
\begin{enumerate}
\item If $\lambda^2 + \mu^2/\gamma < 1$, then for any estimator $\hv:(A,B)\mapsto \hv(A, B)$,
with $\|\hv(A,B)\|_2=\sqrt{n}$, we have $\limsup_{n\to \infty} |\<\hv, v\>| /n = 0$. 
\item  If $\lambda^2 + \mu^2/\gamma > 1$, let $\hv(A,B)$ be 
normalized so that $\|\hv(A,B)\|_2=\sqrt{n}$, and proportional the maximum eigenvector of the matrix $M(\xi_*)$, where
\begin{align}
M(\xi) = A+\frac{2\mu^2}{\lambda^2\gamma^2 \xi}\, B^{\sT}B + \frac{\xi}{2}\, \id_n\, ,
\end{align}
and $\xi_* = \arg\min_{\xi>0}\lambda_{\max}(M(\xi))$. 
Then,  $\liminf_{n\to \infty} |\<\hv, v\>| /n > 0$ in probability.
\end{enumerate}
\end{theorem}
Theorem \ref{thm:graph} is proved by using this threshold result, in conjunction with the universality  Theorem \ref{thm:Universality}.

\section{Related work} \label{sec:related}
The need to incorporate node information in graph clustering has been long recognized. To address the problem, diverse clustering methods have been introduced--- e.g. those based on generative models \cite{newman2016structure,hoff2003random,zanghi2010clustering,yang2009combining,kim2012latent,leskovec2012learning,xu2012model,hoang2014joint,yang2013community}, heuristic model free approaches \cite{binkiewicz2017covariate,zhang2016community,gibert2012graph,zhou2009graph,neville2003clustering, gunnemann2013spectral, dang2012community,cheng2011clustering, silva2012mining, smith2016partitioning}, Bayesian methods \cite{chang2010hierarchical,balasubramanyan2011block} etc. \cite{bothorel2015clustering} surveys other clustering methods for graphs with node and edge attributes. Semisupervised graph clustering \cite{peel2012supervised,eaton2012spin,zhang2014phase}, where labels are available for a few vertices are also somewhat related to our line of enquiry. The literature in this domain is quite vast and extremely diffuse, and thus we do not attempt to provide an exhaustive survey of all related attempts in this direction. 

In terms of rigorous results,
\cite{aicher2014learning,lelarge2015reconstruction} introduced and analyzed a model with informative edges, but they make the strong and unrealistic requirement that the label of individual edges and each of their endpoints are uncorrelated and are only able to prove one side of their conjectured threshold. The papers \cite{binkiewicz2017covariate,zhang2016community} --among others-- rigorously analyze specific heuristics for clustering and provide some guarantees that ensure consistency. However, these results are not optimal. Moreover, it is possible that they only hold in the regime where using either the node covariates or the graph suffices for inference.

Several theoretical works \cite{kanade2016global,mossel2016local} analyze the performance of local algorithms in the semi-supervised setting, i.e., where the true labels are given for a small fraction of nodes. In particular \cite{kanade2016global} establishes that for the two community sparse stochastic block model, correlated recovery is impossible given any vanishing proportion of nodes. Note that this is in stark contrast to Theorem \ref{thm:graph}  (and the Claim for the sparse graph model) above, which posits that given high dimensional covariate information actually shifts the information theoretic threshold for detection and weak recovery. The analysis in \cite{kanade2016global,mossel2016local} is also local in nature, while our algorithms and their analysis go well beyond the diameter of the graph.

\section{Belief propagation: algorithm and cavity prediction}
\label{sec:bp}

Recall the model \eqref{eq:graphmodel}, \eqref{eq:covariatemodel}, where we are given the data $(A^G, B)$ 
and our task is to infer the latent community labels $v$. From a Bayesian perspective, a principled approach computes posterior 
expectation with respect to 
the conditional distribution $\P(v, u | A^G, B)= \P(v, u, A^G, B)/\P(A^G, B)$. This
is, however, not computationally tractable because it
requires to marginalize over $v\in\{+1,-1\}^n$ and $u\in\reals^p$. 
At this point, it becomes necessary to choose an
approximate inference procedure, such as variational inference or mean field approximations \cite{wainwright2008graphical}. In Bayes inference problem on locally-tree like graphs, belief propagation is optimal 
among local algorithms (see for instance \cite{deshpande2015finding} for an explanation of why this is the case). 

The algorithm proceeds by computing, in an iterative fashion \emph{vertex messages} $\eta^t_i, m^t_a$ for $i\in[n]$, $a\in[p]$ and
\emph{edge messages} $\eta^t_{i\to j}$ for all pairs
$(i, j)$ that are connected in the graph $G$. 
For a vertex $i$ of $G$, we denote its neighborhood
in $G$
by $\partial i$.  
Starting from an initialization $(\eta^{t_0}, m^{t_0})_{t_0 = -1, 0}$, we update the messages in the following
\emph{linear} fashion:
\begin{align}
\eta^{t+1}_{i\to j} &= \sqrt{\frac{\mu}{\gamma}} (B^\sT m^t)_i - \frac{\mu}{\gamma} \eta^{t-1}_i + \frac{\lambda}{\sqrt d}\sum_{k\in \partial i \bsl j} \eta^t_{k\to i} - \frac{\lambda \sqrt{d}}{n} \sum_{k \in [n]} \eta^t_k \label{eq:bpedge}, \\
\eta^{t+1}_{i} &= \sqrt{\frac{\mu}{\gamma}} (B^\sT m^t)_i - \frac{\mu}{\gamma} \eta^{t-1}_i + \frac{\lambda}{\sqrt d}\sum_{k\in \partial i } \eta^t_{k\to i} - \frac{\lambda \sqrt{d}}{n} \sum_{k \in [n]} \eta^t_k  \label{eq:bpvertexG},\\
m^{t+1} &= \sqrt{\frac{\mu}{\gamma}} B\eta^t - \mu m^{t-1} \label{eq:bpvertexB}.
\end{align}
Here, and below, we will use $\eta^t= (\eta^t_i)_{i\in [n]}
$, $m^t = (m^t_a)_{a\in[p]}$ to denote
the vectors of vertex messages. After running the algorithm for some number of iterations $t_{\max}$, we return, as an estimate, the sign of the vertex
messages $\eta^{t_{\max}}_i$, i.e.
\begin{align}
 \hv_i(A^G, B) & = \sign(\eta^{t_{\max}}_i).
 \end{align} 
These update equations have
a number of intuitive features. First, in the case
that $\mu = 0$, i.e. we have no covariate information, 
the edge messages become:
\begin{align}
\eta^{t+1}_{i\to j} &=  \frac{\lambda}{\sqrt d}\sum_{k\in \partial i \bsl j} \eta^t_{k\to i} - \frac{\lambda \sqrt{d}}{n} \sum_{k \in [n]} \eta^t_k \label{eq:bpedgenomu}, 
\end{align}
which corresponds closely to the spectral power method on the \emph{nonbacktracking walk} matrix of $G$ \cite{krzakala2013spectral}. Conversely, when $\lambda = 0$, the updates equations on $m^t, \eta^t$ correspond closely to the usual
power iteration to compute singular vectors of $B$. 

We obtain this algorithm from belief propagation
using two approximations. First, we linearize
the belief propagation update equations around a certain
`zero information' fixed point. Second, we use
an `approximate message passing' version of the belief propagation
updates which results in the addition of the memory
terms in Eqs. \eqref{eq:bpedge}, \eqref{eq:bpvertexG}, \eqref{eq:bpvertexB}. The details of these approximations
are quite standard and deferred to Appendix \ref{sec:bpderivation}. 
For a heuristic discussion, we refer the interested reader to the tutorials \cite{MontanariChapter,tramel2014statistical}
(for the Gaussian approximation)
and the papers \cite{decelle2011asymptotic,krzakala2013spectral}
(for the linearization procedure).

As with belief propagation, the behavior of this iterative algorithm, in the limit
$p, n\to \infty$ can be tracked using a distributional 
recursion called \emph{density evolution}. 

\begin{definition}[Density evolution]\label{def:densityevolutiondef}
Let $(\mbar,U)$ and $(\etabar,V)$ be independent random vectors such that  $U \sim \normal(0, 1)$, $V\sim\Unif(\{\pm 1\})$,
$\mbar, \etabar$ have finite variance. 
Further assume that $(\etabar,V)\ed (-\etabar,-V)$ and
$(\mbar,U)\ed (-\mbar,-U)$ (where $\ed$ denotes equality in distribution).

We then define new random pairs $(\mbar',U')$ and $(\etabar',V')$,
where $U' \sim \normal(0, 1)$, $V'\sim\Unif(\{\pm 1\})$, and
$(\etabar,V)\ed (-\etabar,-V)$, $(\mbar,U)\ed (-\mbar,-U)$,
via the following distributional equation
\begin{align}
\mbar'\big|_{U'} &\stackrel{\d}{=} \mu \E\{V\etabar\} U' + \big( \mu \E\{\etabar^2\}\big)^{1/2} \zeta_1 , \\
\etabar'\big|_{V'=+1} & 
\stackrel{\d}{=} \frac{\lambda}{\sqrt{d}} \Big[
\sum_{k=1}^{k_+} \etabar_k\big|_{+} + \sum_{k=1}^{k_-} 
\etabar_k\big|_{-} \Big] - \lambda \sqrt{d} \E\{ \etabar\} \nonumber\\
& +\frac{\mu}{\gamma} \E\{ U\mbar\} + \Big(\frac{\mu}{\gamma} \E\{\mbar^2\}\Big) ^{1/2}  \zeta_2.
\end{align}
Here we use the notation $X|_{Y} \ed Z$ to mean that the conditional distribution of $X$ given $Y$ is the same as the (unconditional) distribution of $Z$. Notice that the distribution of $\etabar'\big|_{V'=-}$ is determined by the last 
equation using the symmetry property. Further $\etabar_k|_+$
and $\etabar_k|_-$ denote independent random variables distributed (respectively) as $\etabar|_{V=+}$ and $\etabar|_{V=-}$.
Finally $k_+ \sim \Poiss(d/2 + \lambda \sqrt{d}/2)$, 
$k_- \sim \Poiss(d/2 - \lambda \sqrt{d}/2)$, $\zeta_1\sim\normal(0, 1)$ and $\zeta_2\sim\normal(0, 1)$ 
are mutually independent and independent from the previous random variables. 

The \emph{density evolution} map, 
denoted by $\DE$, is defined as the mapping from the law of  $(\etabar,V, \mbar,U)$ to the law of
$(\etabar',V', \mbar',U')$. With a slight abuse
of notation, we will omit $V,U$, $V',U'$, whose distribution is
left unchanged and write 
\begin{align}
(\etabar', \mbar')=\DE (\etabar,  \mbar )\, . 
\end{align}
\end{definition}

The following claim is the core of the cavity prediction.
It states that the density evolution recursion faithfully
describes the distribution of the iterates $\eta^t, m^t$.
\begin{claim}
\label{claim:DEandBP}
Let $(\etabar^0,V)$, $(\mbar^0,U)$ be random vectors satisfying the conditions of definition \ref{def:densityevolutiondef}. Define the density evolution sequence 
$(\etabar^t, \mbar^t) = \DE^t(\etabar^0, \mbar^0)$, 
i.e. the result of iteratively applying the mapping $\DE$ $t$ times.

Consider the linear message passing algorithm of Eqs.~(\ref{eq:bpedge}) to (\ref{eq:bpvertexB}), with the following initialization. We set $(m^0_r)_{r\in [p]}$ conditionally independent given $u$, with conditional distribution
$m^0_r|_u \ed \mbar^0|_{U = \sqrt{p}u_r}$. Analogously, 
$\eta^0_i, \eta_{i\to j}^0$ are conditionally independent given $v$
with $\eta^0_i|_v \ed \etabar^0|_{V = v_i}$, 
$\eta^0_{i\to j}|_v \ed \etabar^0|_{V = v_i}$. Finally 
$\eta^{-1}_i = \eta^{-1}_{i\to j} = m^{-1}_r = 0$ for all $i,j, r$.

Then, as $n, p\to \infty$ with $p/n \to 1/\gamma$, the following holds for uniformly random indices $i \in [n]$ and $a\in [p]$:
\begin{align}
(m^t_a, u_a \sqrt{p}) & \convD (\mbar^t, U)\\
(\eta^t_i, v_i) &\convD (\etabar^t, V).
\end{align}
\end{claim}

The following simple lemma shows the instability 
of the density evolution recursion. 
\begin{lemma}\label{lem:densityevolutioninstability}
Under the density evolution mapping, 
we obtain the random variables
$(\etabar', \mbar' )= \DE(\etabar, \mbar'$
Let $\mom$ and $\mom'$ denote 
the vector of the first two moments of $(\etabar,V,\mbar,U)$ and $(\etabar',V',\mbar',U')$ defined as follows:
\begin{align}
\mom &= (\E\{V\etabar\}, \E\{U\mbar\}, 
\E\{\etabar^2\}, \E\{\mbar^2\})\, ,
\end{align}
and similarly for $\mom'$. Then, for $\|\mom\|_2\to 0$, we have
\begin{align}
\mom' = \left[\begin{matrix}
\lambda^2 & \mu/\gamma & 0 & 0\\
\mu & 0 & 0 & 0 \\
0& 0 & \lambda^2 & \mu/\gamma \\
0 & 0 & \mu & 0 \\
\end{matrix}\right]\mom + O(\|\mom\|^2)
\end{align}
In particular, the linearized map $\mom\mapsto \mom'$ 
at $\mom=0$ has spectral radius larger than one if and only if 
$\lambda^2+\mu^2/\gamma>1$.
\end{lemma}

The interpretation of the lemma is as follows. If we choose an initialization $(\etabar^0, V)$, $(\mbar^0,U)$ with $\etabar^0,\mbar^0$ \emph{positively correlated} with $V$ and $U$,
then this correlation increases exponentially over 
time if and only if $\lambda^2+\mu^2/\gamma>1$\footnote{Notice 
that both the messages variance $\E(\eta^2)$ and 
covariance with the ground truth $\E(\eta V)$ increase, but the normalized correlation (correlation divided by standard deviation)
increases.}. In other words, a small initial correlation is amplified.

While we do not have an initialization that is positively correlated with the true labels,  a random initialization $\eta^0, m^0$ 
has a random correlation with $v, u$ of order
$1/\sqrt{n}$.  If $\lambda^2+\mu^2/\gamma>1$, this correlation
is amplified over iterations, yielding  a nontrivial
reconstruction of $v$. On the other hand, if $\lambda^2+\mu^2/\gamma<1$ then this correlation is expected to remain small, indicating that the algorithm does not yield a useful estimate. 


\section{Proof overview}\label{sec:overview}

 As mentioned above, a key step of our analysis is provided by Theorem \ref{thm:gaussian},
 which establishes a weak recovery threshold for the Gaussian observation model
 of Eqs.~(\ref{eq:gaussiangraphmodel}), (\ref{eq:gaussiancovariatemodel}). 
 
 The  proof proceeds in two steps: first, we prove that, for $\lambda^2+\mu^2/\gamma<1$ it is impossible  
 to distinguish between data $A, B$ generated according to this model, and data generated according to the
 null model $\mu=\lambda=0$. Denoting by $\prob_{\lambda,\mu}$ the law of data $A, B$, this is proved via a standard
 second moment argument. Namely, we bound the chi square distance uniformly in $n,p$
 \begin{align}
 \chi^2(\prob_{\lambda,\mu},\prob_{0,0}) \equiv \E_{0,0}\left\{\left(\frac{\de \prob_{\lambda,\mu}}{\de\prob_{0,0}}\right)^2\right\}-1\le C\, ,
 \end{align}
 and then bound the total variation distance by the chi-squared
 distance $\|\prob_{\lambda,\mu}-\prob_{0,0}\|_{TV}\le 1-(\chi^2(\prob_{\lambda,\mu},\prob_{0,0})+1)^{-1}$.
 This in turn implies that no test can distinguish between the two hypotheses with probability approaching one as $n,p\to\infty$. 
 The chi-squared bound also allows to show that weak recovery is impossible in the same regime.
 
In order to prove that weak recovery is possible for  $\lambda^2+\mu^2/\gamma> 1$, we consider the
following optimization problem over $x\in\reals^n$, $y\in\reals^p$:
\begin{align}
\mbox{maximize} \, & \;\; \<x,Ax\> + b_*\<x,By\>,\label{eq:OptEstimator}\\
\mbox{subject to}\, & \;\; \|x\|_2 = \|y\|_2 =1\, .
\end{align}
where $b_* = \frac{2\mu}{\lambda \gamma}$.
Denoting solution of this problem by $(\hx, \hy)$, we output the (soft) label estimates $\hv = \sqrt{n} \hx$. 
This definition turns out to be equivalent to the spectral algorithm in the statement
of Theorem \ref{thm:gaussian}, and is therefore efficiently computable.

This optimization problem undergoes a phase transition exactly at the weak recovery threshold $\lambda^2+\mu^2/\gamma=1$,
as stated below.
\begin{lemma}\label{lemma:upper_bound}
Denote by $T = T_{n,p}(A,B)$ the value of the optimization problem (\ref{eq:OptEstimator}). 
\begin{enumerate}
\item[(i)] If $\lambda^2 + \frac{\mu^2}{\gamma}<1$, then, almost surely 
\begin{align}
\lim_{n,p\to \infty}T_{n,p}(A,B) = 2\sqrt{1 + \frac{b_*^2 \gamma}{4}} + b_*\, .
\end{align}
\item[(ii)] If $\lambda,\mu>0$, and  $\lambda^2 + \frac{\mu^2}{\gamma}>1$ then there exists $\delta= \delta(\lambda,\mu)>0$ such that, almost surely
\begin{align}
\lim_{n,p\to \infty}T_{n,p}(A,B) = 2\sqrt{1 + \frac{b_*^2 \gamma}{4}} + b_* +\delta(\lambda,\mu)\, .
\end{align}
\item[(iii)] Further, define 
\begin{align}
\tilde{T}_{n,p}(\tilde{\delta};A,B) = \sup_{\|x\|= \|y\|=1, | \langle x, v \rangle | < \tilde{\delta} \sqrt{n}} \Big[ \langle x, Ax \rangle + b_* \langle x, By \rangle \Big]. \nonumber
\end{align}
Then for each $\delta >0$, there exists $\tilde{\delta}>0$ sufficiently small, such that, amlost surely
\begin{align}
\lim_{n,p\to\infty}\tilde{T}_{n,p} (\tilde{\delta};A,B) < 2 \sqrt{1 + \frac{b_*^2 \gamma}{4}} + b_* + \frac{\delta}{2} \, .
\end{align} 
\end{enumerate}
\end{lemma}
The first two points imply that $T_{n,p}(A,B)$ provide a statistic to  distinguish
between $\prob_{0,0}$ and $\prob_{\lambda,\mu}$ with probability of error that vanishes as $n,p\to\infty$
if $\lambda^2+\mu^2/\gamma>1$. The third point (in conjunction with the second one) guarantees that
the maximizer $\hx$ is positively correlated with $v$, and hence implies weak recovery.

In fact, we prove a stronger result that provides an asymptotic expression for the
value $T_{n,p}(A,B)$ for all $\lambda,\mu$. We obtain the above phase-transition result by specializing 
the resulting formula in the two regimes $\lambda^2+\mu^2/\gamma<1$ and $\lambda^2+\mu^2/\gamma>1$.
We prove this asymptotic formula by Gaussian process comparison, using Sudakov-Fernique inequality. 
Namely, we compare the Gaussian process appearing in the optimization problem of
Eq.~(\ref{eq:OptEstimator}) with the following ones:
\begin{align}
{\mathcal X}_1(x,y) & = 
\frac{\lambda}{n} \<x, v_0\>^2 + \<x, \tg_x\>
+ b_* \sqrt{\frac{\mu}{n}} \<x, v_0\> \<y, u_0\> + \<y, \tg_y\> \, ,\\
{\mathcal X}_2(x,y) & = 
\frac{\lambda}{n} \<x, v_0\>^2 + \frac{1}{2} \<x, \tW_x x\> 
+ b_*\sqrt{\frac{\mu}{n}} \<x, v_0\> \<y, u_0\> + \frac{1}{2}\<y, \tW_y y\> \, ,
\end{align}
where $\tg_x$, $\tg_y$ are isotropic Gaussian vectors, with suitably chosen variances, and 
$\tW_x$, $\tW_y$ are GOE matrices, again with properly chosen variances. We prove that
$\max_{x,y}\cX_1(x,y)$ yields an upper bound on $T_{n,p}(A,B)$, and $\max_{x,y}\cX_2(x,y)$
yields a lower bound on the same quantity.

Note that maximizing the first process ${\mathcal X}_1(x,y)$ essentially reduces to solving a separable problem over
the coordinates of $x$ and $y$ and hence to an explicit expression. On the other hand, maximizing
the second process leads (after decoupling the term $\<x, v_0\> \<y, u_0\>$) to two separate problems,
one for the vector $x$, and the other for $y$. Each of the two problems reduce to finding the maximum eigenvector
of a rank-one deformation of a GOE matrix, a problem for which we can leverage on significant amount of information
from random matrix theory.
The resulting upper and lower bound coincide asymptotically.

As is often the case with Gaussian comparison arguments, the proof is remarkably compact, and somewhat surprising
(it is unclear a priori that the two bounds should coincide asymptotically). While upper bounds by processes
of the type of $\cX_1(x,y)$ are quite common in random matrix theory, we think that the lower bound by $\cX_2(x,y)$ (which  is crucial for proving our main theorem) is novel and might have interesting generalizations.


\section{Experiments}\label{sec:expt}

\def\dirE{{\overrightarrow{E}}}
\def\vBP{{\widehat{v}^{\sf BP}}}
\def\uBP{{\widehat{u}^{\sf BP}}}

\begin{figure}[t]
\centering
\includegraphics[width=0.32\linewidth]{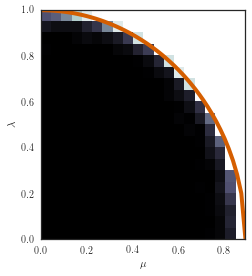}
\includegraphics[width=0.32\linewidth]{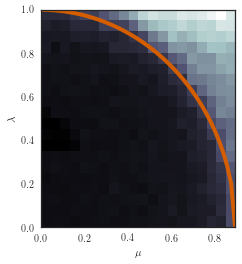}
\includegraphics[width=0.32\linewidth]{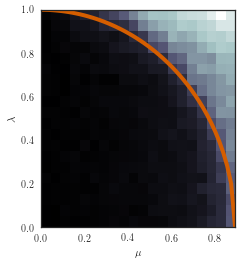}
\caption{(Left) Empirical probability of rejecting the null
(lighter is higher) using BP test. 
(Middle) Mean overlap $|\<\vBP, v\>/n|$ and 
(Right) mean covariate overlap $|\<\uBP, u\>|$ attained by BP
estimate.
\label{fig:phase_transition} }
\end{figure}

We demonstrate the efficacy of the full belief
propagation algorithm, restated below:
\begin{align}
\eta^{t+1}_{i} & = \sqrt{\frac{\mu}{\gamma}} \sum_{q\in [p]}
B_{qi}m^t_q -%
 \frac{\mu}{\gamma}\bigg( \sum_{q\in [p]}\frac{ B_{qi}^2}{\tau^t_q}\bigg) \tanh(\eta^{t-1}_i) + %
\sum_{k\in \partial i} f(\eta^t_{k\to i}; \rho) - \sum_{k\in[n]} f(\eta^t_k; \rho_n) \, ,\label{eq:textamp1}\\
\eta^{t+1}_{i\to j} & =  \sqrt{\frac{\mu}{\gamma}} \sum_{q\in [p]}
B_{qi}m^t_q -%
 \frac{\mu}{\gamma}\bigg( \sum_{q\in [p]}\frac{ B_{qi}^2}{\tau^t_q}\bigg) \tanh(\eta^{t-1}_i) + %
\sum_{k\in \partial i\setminus j} f(\eta^t_{k\to i}; \rho) - \sum_{k\in[n]} f(\eta^t_k; \rho_n)\, ,\label{eq:textamp2} \\
m^{t+1}_q & = \frac{\sqrt{\mu/\gamma}}{\tau^{t+1}_q}%
 \sum_{j\in [n]} B_{qj}\tanh(\eta^t_j)- %
 {\frac{\mu}{\gamma\tau^{t+1}_q}} %
 \bigg(\sum_{j\in [n]} B_{qj}^2 \sech^2(\eta^t_j) \bigg)%
  m^{t-1}_q \label{eq:textamp3} \\
\tau^{t+1}_q & = \left(1 + \mu - \frac{\mu}{\gamma} \sum_{j\in [n]} B_{qj}^2 \sech^2(\eta^t_j)\right)^{-1}\label{eq:textamp4}.
\end{align}
Here the function $f(; \rho)$ and the parameters $\rho, \rho_n$ are
defined as:
\begin{align}
f(z; \rho) &\equiv \frac{1}{2}\log \Big(\frac{\cosh(z + \rho)}{\cosh(z-\rho)}\Big)\, , \\
\rho &\equiv \tanh^{-1}(\lambda/\sqrt{d})\, , \\
\rho_n & \equiv \tanh^{-1}\Big(\frac{\lambda \sqrt{d}}{n-d}\Big). 
\end{align}
We refer the reader to Appendix \ref{sec:bpderivation} for a derivation
of the algorithm. As demonstrated in Appendix \ref{sec:bpderivation}, the
BP algorithm in Section \ref{sec:bp} is obtained by linearizing the above
in $\eta$. 

In our experiments, we perform 100 Monte Carlo runs of the following process:
\begin{enumerate}
\item Sample $A^G, B$ from $\P_{\lambda, \mu}$ with $n = 800, p=1000, d = 5$.
\item Run BP algorithm for $T = 50$ iterations with random initialization
$\eta^0_i, \eta^{-1}_i, m^0_a, m^{-1}_a \sim_\textrm{iid} \normal(0, 0.01)$.
yielding vertex and covariate iterates $\eta^T \in \reals^n$, $m^T\in \reals^p$.
\item Reject the null hypothesis if $\norm{\eta^T}_2 > \norm{\eta^0}_2$, else accept the null.
\item Return estimates $\vBP_i = \sign(\eta^T_i)$, $\uBP_a = m^T_a/\norm{m^T}_2$.
\end{enumerate}

Figure \ref{fig:phase_transition} (left) shows empirical probabilities
of rejecting the null for 
$(\lambda, \mu)\in [0, 1]\times[0, \sqrt{\gamma}]$. The next two
plots display the mean overlap $|\<\vBP, v\>/n|$ and
$\<\uBP, u\>/\norm{u}$ achieved by the BP estimates (lighter
is higher overlap). Below the theoretical curve
(red) of $\lambda^2 + \mu^2/\gamma = 1$,
the null hypothesis is accepted and the estimates show negligible
correlation with the truth. These results are in excellent 
agreement with our theory.

\section*{Acknowledgements}

A.M. was partially supported by grants NSF DMS-1613091, NSF CCF-1714305 and NSF IIS-1741162. E.M was partially supported
by grants NSF DMS-1737944 and ONR N00014-17-1-2598. Y.D would
like to acknowledge Nilesh Tripuraneni for discussions about
this paper. 

\iftoggle{nips}{
\newpage

\bibliographystyle{amsalpha}
\bibliography{all-bibliography}
\newpage
\appendix
\section{Proof of Theorem \ref{thm:gaussian}}
We establish Theorem \ref{thm:gaussian} in this section. First, we introduce the notion of contiguity of measures 
\begin{definition}
Let $\{P_n\}$ and $\{Q_n\}$ be two sequences of probability measures on the measurable space $(\Omega_n , \mathcal{F}_n)$. We say that $P_n$ is contiguous to $Q_n$ if for any sequence of events $A_n$ with $Q_n(A_n) \to 0$, $P_n(A_n) \to 0$. 
\end{definition}
It is standard that for two sequences of probability measures $P_n$ and $Q_n$ with $P_n$ contiguous to $Q_n$, $\limsup_{n \to \infty }d_{\mathrm{TV}}(P_n, Q_n) <1 $. The following lemma provides sufficient conditions for establishing contiguity of two sequence of probability measures. 
\begin{lemma}[see e.g. \cite{montanari2015limitation} ]
\label{lem:contiguity_suff}
Let $P_n$ and $Q_n$ be two sequences of probability measures on $(\Omega_n, \mathcal{F}_n)$. Then $P_n$ is contiguous to $Q_n$ if 
\begin{align}
\E_{Q_n}\Big[ \Big( \frac{\de P_n}{\de Q_n} \Big)^2  \Big] \nonumber
\end{align}
exists and remains bounded as $n \to \infty$. 
\end{lemma}
Our next result establishes that asymptotically error-free detection is impossible below the conjectured detection boundary. 
\begin{lemma}\label{lem:detection_lowerbdd}
Let $\lambda, \mu >0$ with $\lambda^2 + \frac{\mu^2}{\gamma} <1$. Then  $\P_{\lambda,\mu}$ is contiguous to $\P_{0,0}$. 
\end{lemma}
To establish that consistent detection is possible above this boundary, we need the following lemma. Recall the matrices $A,B$ from the Gaussian model \eqref{eq:gaussiangraphmodel}, \eqref{eq:gaussiancovariatemodel}. 
\begin{lemma}\label{lemma:upper_bound}
Let $b_* = \frac{2\mu}{\lambda \gamma}$. Define 
\begin{align}
T = \sup_{\|x\| = \|y \| =1} \Big[ \langle x, Ax \rangle + b_* \langle x, By \rangle \Big]. \nonumber
\end{align}
\begin{enumerate}
\item[(i)] Under $\P_{0,0}$, as $n, p \to \infty$, $T \to 2 \sqrt{1 + \frac{b_*^2 \gamma}{4}} + b_*$ almost surely.  

\item[(ii)] Let $\lambda, \mu >0$, $\varepsilon>0$, with $\lambda^2 + \frac{\mu^2}{\gamma}>1 + \varepsilon $. Then as $n, p \to \infty$, 
\begin{align}
\P_{\lambda,\mu}\Big(T > 2 \sqrt{1 + \frac{b_*^2 \gamma}{4}} + b_* + \delta \Big) \to 1,  \nonumber
\end{align}
where $\delta := \delta(\varepsilon)>0$. 

\item[(iii)] Further, define 
\begin{align}
\tilde{T}(\tilde{\delta}) = \sup_{\|x\|= \|y\|=1, 0< \langle x, v \rangle < \tilde{\delta} \sqrt{n}} \Big[ \langle x, Ax \rangle + b_* \langle x, By \rangle \Big]. \nonumber
\end{align}
Then for each $\delta >0$, there exists $\tilde{\delta}>0$ sufficiently small, such that  as $n,p \to \infty$,
\begin{align}
\P_{\lambda, \mu} \Big( \tilde{T} (\tilde{\delta}) < 2 \sqrt{1 + \frac{b_*^2 \gamma}{4}} + b_* + \frac{\delta}{2} \Big) \to 1. \nonumber
\end{align}
\end{enumerate}
\end{lemma}
We defer the proofs of Lemma \ref{lem:detection_lowerbdd} and Lemma \ref{lemma:upper_bound} to Sections \ref{sec:proof_lowerbdd} and Section \ref{sec:proof_upperbdd} respectively, and complete the proof of Theorem \ref{thm:gaussian}, armed with these results. 

\begin{proof}[Proof of Theorem \ref{thm:gaussian}]
The proof is comparatively straightforward, once we have Lemma \ref{lem:detection_lowerbdd} and \ref{lemma:upper_bound}. Note that Lemma \ref{lem:detection_lowerbdd} immediately implies that $\P_{\lambda,\mu}$ is contiguous to $\P_{0,0}$ for $\lambda^2 + \frac{\lambda^2}{\gamma}<1$. 

Next, let $\lambda, \mu>0$ such that $\lambda^2 + \frac{\mu^2}{ \gamma} > 1 + \varepsilon$ for some $\varepsilon>0$. In this case, consider the test which rejects the null hypothesis $\mathrm{H}_0$ if $T > 2 \sqrt{1 + \frac{b_*^2 \gamma}{4}} + b_* + \delta$. Lemma \ref{lemma:upper_bound} immediately implies that the Type I and II errors of this test vanish in this setting. 

Finally, we prove that weak recovery is possible whenever $\lambda^2 + \frac{\mu^2}{\gamma} >1$. To this end, let $(\hat{x}, \hat{y})$ be the maximizer of $\langle x, Ax \rangle + b_* \langle y, B x \rangle$, with $\|x \| = \| y \| =1$. Combining parts $(ii)$ and $(iii)$ of Lemma \ref{lemma:upper_bound}, we conclude that $\hat{x}$ achieves weak recovery of the community assignment vector. 
\end{proof}

\subsection{Proof of Lemma \ref{lem:detection_lowerbdd}}
\label{sec:proof_lowerbdd}
Fix $\lambda, \mu >0$ satisfying $\lambda^2 + \frac{\mu^2}{\gamma} <1$. 
We start with the likelihood,
\begin{align}
&L(u,v) = \frac{\de \P_{\lambda, \mu}}{\de \P_{0,0}} = L_1 (u,v) L_2 (u,v) , \nonumber\\
&L_1(u,v) = \exp\Big[ \frac{\lambda}{2} \langle  A, v v^{T} \rangle - \frac{\lambda^2 n}{4}\Big]. \label{eqgoe}\\
&L_2(u,v) = \exp\Big[ p \sqrt{\frac{\mu}{n}} \langle B, uv^{T} \rangle - \frac{\mu p}{2 } \|u\|^2 \Big]. \label{eqwishart}
\end{align}
We denote the prior joint distribution of (u,v) as $\pi$, and set 
\begin{align}
L_{\pi} = \E_{(u,v) \sim \pi } \Big[ L(u,v) \Big]. \nonumber
\end{align}
To establish contiguity, we bound the second moment of $L_{\pi}$ under the null hypothesis, and appeal to Lemma \ref{lem:contiguity_suff}. In particular, we denote $\E_{0}[\cdot]$ to be the expectation operator under the distribution $P_{(0,0)}$ and compute 
\begin{align}
\E_{0}[ L_{\pi}^2 ] = \E_{0}[ \E_{(u_1, v_1), (u_2, v_2)}\Big[ L(u_1, v_1) L(u_2 , v_2) \Big] \Big] = \E_{(u_1, v_1), (u_2, v_2)} \Big[ \E_{0} \Big[ L(u_1, v_1) L(u_2 , v_2) \Big] \Big], \nonumber 
\end{align}
where $(u_1, v_1) , (u_2, v_2)$ are i.i.d. draws from the prior $\pi$, and the last equality follows by Fubini's theorem. We have, using \eqref{eqgoe} and \eqref{eqwishart}, 
\begin{align}
&L(u_1, v_1) L( u_2, v_2) \nonumber\\
&= \exp\Big[ - \frac{\lambda^2 n}{2} - \frac{\mu p}{2 n} \Big(\| u_1 \|^2 + \| u_2 \|^2 \Big)  + \frac{\lambda}{2} \Big\langle A, v_1 v_1^T + v_2 v_2^T \Big\rangle + p\sqrt{\frac{\mu}{n}} \Big\langle B, u_1 v_1^T + u_2 v_2^T  \Big\rangle \Big] \nonumber. 
\end{align}
Taking expectation under $\E_{0}[\cdot]$, upon simplification, we obtain,
\begin{align}
\E_0[L_{\pi}^2 ] &= \E_{(u_1, v_1),(u_2,v_2)}\Big[ \exp\Big[ \frac{\lambda^2}{2n} \langle v_1 , v_2 \rangle^2  + \frac{\mu p}{n} \langle u_1, u_2 \rangle \langle v_1, v_2 \rangle\Big] \Big]\\
&=\E_{(u_1, v_1),(u_2,v_2)}\Big[\exp \Big[n \Big( \frac{\lambda^2}{2} \Big(\frac{\<v_1, v_2\>}{n} \Big)^2 + \frac{\mu}{\gamma} \< u_1, u_2 \> \frac{\<v_1, v_2\>}{n}   \Big) \Big]\Big] \\
&=\E\Big[\exp \Big[n \Big( \frac{\lambda^2}{2} X^2 + \frac{\mu}{\gamma} XY  \Big) \Big]\Big]
\label{eq:sec_moment}
\end{align}

Here that $X, Y\in [-1,+1]$ are independent, with $X$ distributed as the normalized sum of $n$ Radamacher random variables, and $Y$ as the first coordinate of a uniform vector on the unit sphere. In particular, defining $h(s) = -((1+s)/2)\log((1+s))-((1-s)/2)\log((1-s))$, and denoting by $f_Y$ the density of $Y$, 
we have, for $s\in (2/n){\mathbb Z}$ 
\begin{align}
\prob\big(X= s\big) &= \frac{1}{2^n}\binom{n}{n(1+s/2)}\\
&\le \frac{C}{n^{1/2}}\, e^{n h(s)}\\
f_Y(y) & = \frac{\Gamma(p/2)}{\Gamma((p-1)/2)\Gamma(1/2)}
(1-y^2)^{(p-3)/2}\\
& \le C\sqrt{n} (1-y^2)^{p/2}\, .
\end{align}
Approximating sums by integrals, and using $h(s)\le -s^2/2$, we get
\begin{align}
\E_0[L_{\pi}^2 ]&\le C n\int_{[-1,1]^2}
\exp\Big\{ n\Big[\frac{\lambda^2}{2} s^2+ \frac{\mu}{\gamma} s y+
h(s) +\frac{1}{2\gamma}\log(1-y^2) \Big\} \de s \de y\\
&\le Cn \int_{{\mathbb R}^2} 
\exp\Big\{ n\Big[\frac{\lambda^2}{2} s^2+ \frac{\mu}{\gamma} s y
-\frac{s^2}{2} -\frac{y^2}{2\gamma} \Big]\Big\} \de s \de y
\le C'\, .
\end{align}
The last step holds for $\lambda^2+\mu^2/\gamma<1$.

Next, we turn to the proof of Lemma \ref{lemma:upper_bound}. This is the main technical contribution of this paper, and uses a novel Gaussian process comparison argument based on 
Sudakov-Fernique comparison. 
\subsection{A Gaussian process comparison result}
Let $Z\sim\reals^{p \times n}$ and $W\sim\reals^{n\times n}$ denote random matrices
with independent entries as follows. 
\begin{align}\label{eq:model}
 W_{ij} &\sim \begin{cases}
 \normal(0, \rho/n) &\text{ if } i < j \\
 \normal(0, 2\rho/n) &\text{ if } i = j
 \end{cases} \\
 \text{ where }W_{ij} &= W_{ji}, \nonumber \\
 Z_{ai} &\sim \normal(0, \tau/p). 
 \end{align} 

 For an integer $N > 0$,  we let $\sphere^{N}$ denote the sphere of radius $\sqrt{N}$ 
 in $N$ dimensions, i.e. $\sphere^N = \{x\in\reals^N: \norm{x}_2^2 = N\}$.
Furthermore let $u_0 \in \sphere^p$ and $v_0\in \{\pm 1\}^n$ be fixed 
vectors. We denote the standard inner product between vectors $x, y \in\reals^N$ as $\<x, y\> = \sum_{i} x_i y_i$.
The normalized version will be useful as well: we define $\<x, y\>_N \equiv \sum_i x_iy_i/N$. 

 We are interested in characterizing the behavior of the following optimization
 problem in the limit high-dimensional limit $p, n\to \infty$ with constant
 aspect ratio $n/p = \gamma \in (0, \infty)$. 

 \begin{align}\label{eq:basictest}
\OPT(\lambda, \mu, b)  &\equiv \frac{1}{n} \E \max_{(x, y) \in\sphere^{n} \times \sphere^{p}}
\Big[ \big( \frac{\lambda}{n} \< x, v_0\>^2 + \<x, W x\>\big) +  b \,\big(\sqrt{\frac{\mu}{np}}\<x, v_0\>\, \<y, u_0\> + \<y, Z x\>
\big) \Big] . \nonumber
 \end{align}

We  now introduce two different comparison processes which
give upper and lower bounds to 
$\OPT(\lambda, \mu, b)$. Their asymptotic values will coincide in the high dimensional limit $n,p \to \infty$ with $n/p = \gamma$. Let $g_x$, $g_y$, $W_x$ and $W_y$
be:
\begin{align}
 g_x &\sim \normal(0, (4\rho + b^2\tau) \id_n) \\
 g_y  &\sim \normal(0,  b^2 \tau n/ p \id_p), \\
 (W_x)_{ij} &\sim \begin{cases}
 \normal(0, (4 \rho + b^2 \tau) /n ) &\text{ if } i < j \\
 \normal(0, 2(4 \rho + b^2 \tau)/n) &\text{ if } i = j
 \end{cases} \\
 (W_y)_{ij} &\sim \begin{cases}
 \normal(0, b^2 \tau n/p^2 ) &\text{ if } i < j \\
 \normal(0, 2 b^2 \tau n/p^2) &\text{ if } i = j
 \end{cases}
\end{align}

\begin{proposition}\label{prop:slepianupper}
We have
\begin{align}
\OPT(\lambda, \mu, b) &\le \frac{1}{n} \E \max_{(x, y) \in \sphere^n\times\sphere^p} 
\frac{\lambda}{n} \<x, v_0\>^2 + \<x, g_x\>
+ b \sqrt{\frac{\mu}{np}} \<x, v_0\> \<y, u_0\> + \<y, g_y\> \nonumber\\
\OPT(\lambda, \mu, b) &\ge \frac{1}{n} \E \max_{(x, y) \in \sphere^n\times\sphere^p} 
\frac{\lambda}{n} \<x, v_0\>^2 + \frac{1}{2} \<x, W_x x\> 
+ b\sqrt{\frac{\mu}{np}} \<x, v_0\> \<y, u_0\> + \frac{1}{2}\<y, W_y y\> \label{eq:lower}
\end{align}
\end{proposition}
\begin{proof}
The proof is via Sudakov-Fernique inequality. First we compute the 
distances induced by the three processes. For any pair $(x, y), (x', y')$: 
\begin{align}
&\frac{1}{4n}\big(\E\{ (\<x, Wx\> + b \<y, Zx\> -   \<x', Wx'\> - b\<y', Zx'\>)^2   \} \big)
 = \rho  ( 1 - \<x, x'\>_n^2) + \frac{b^2\tau}{2}   (1 - \<x, x'\> _n\<y, y'\>_p) \nonumber\\
&\frac{1}{n} \big(\E\{ ( \<x, g_x\> +\<y, g_y\>  - \<x', g_x\>  - \<y', g_y\>)^2\} \big) 
 =  2( 4 \rho + b^2 \tau )  (1- \<x, x'\>_n) + 2 b^2 \tau   ( 1 - \<y, y'\>_p) \nonumber \\
&\frac{1}{4n} \big(\E\{ (\<x, W_x x\> + \<y, W_y y\> - \<x', W_x x\> - \<y', W_y y'\>)^2  \} \big)
 =  (  \rho + \frac{b^2 \tau}{4} ) (1 - \<x, x'\>_n^2) + \frac{b^2  \tau}{4}  (1-\<y, y'\>_p^2). \nonumber
\end{align}
This immediately gives:
\begin{align}
&\frac{1}{n}\big(\E\{ (\<x, Wx\> + b \<y, Zx\> -   \<x', Wx'\> - b \<y', Zx'\>)^2   \} \big) - \nonumber\\
&\frac{1}{n} \big(\E\{ ( \<x, g_x\> +\<y, g_y\>  - \<x', g_x\>  - \<y, g_y'\>)^2\} \big) \nonumber \\
&=  -4 \rho  ( 1 - \<x, x'\>_n)^2 - 2{ b^2 \tau}   (1 - \<x, x'\> _n)(1-\<y, y'\>_p)\le 0, \nonumber\\
&\frac{1}{4n}\big(\E\{ (\<x, Wx\> + b \<y, Zx\> -   \<x', Wx'\> - b \<y', Zx'\>)^2   \} \big) - \nonumber\\
&\frac{1}{4n} \big(\E\{ (\<x, W_x x\> + \<y, W_y y\> - \<x', W_x x\> - \<y', W_y y'\>)^2  \} \big) \nonumber \\
&= \frac{b^2 \tau}{4} (\<x, x'\>_n - \<y, y'\>_p)^2 \ge 0.\nonumber
\end{align}
The claim follows. 
\end{proof}

An immediate corollary of this is the following tight 
characterization for the null value, i.e. the case when $\mu = \lambda = 0$:
\begin{corollary}\label{cor:zerosignal}
For any $\rho, \tau$ as $n, p$ diverge with $n/p\to \gamma$, we have
\begin{align}
\lim_{n\to\infty} \OPT(0, 0) & = \sqrt{4 \rho + b^2  \tau} + b \sqrt{\frac{\tau}{\gamma}}
\end{align}
\end{corollary}

Note that this upper bound generalizes the maximum eigenvalue and singular value 
bounds of $W$, $Z$ respectively. In particular, the case $\tau = 0$ corresponds
to the maximum eigenvalue of $W$, which yields $\OPT = 2\sqrt{\rho}$ while the maximum singular value
of $Z$ can be recovered by setting $\rho$ to 0 and $b$ to 1, yielding 
$\OPT = \sqrt{\tau} (1+\gamma^{-1/2})$. Corollary \ref{cor:zerosignal} demonstrates the limit
for the case when $\mu = \lambda = 0$. The following theorem
gives the limiting value when $\lambda, \mu$ may be nonzero. 
\begin{theorem}\label{thm:main_comparison}
Suppose $\sG: \reals\times\reals_+\to\reals$ is
as follows:
\begin{align}
\sG(\kappa, \sigma^2) & = \begin{cases}
{\kappa}/{2} + {\sigma^2}/{2\kappa} &\text{ if } \kappa^2 \ge \sigma^2, \\
\sigma &\text{ otherwise.}
\end{cases} \label{def:G}
\end{align}
Then the optimal value $\OPT(\lambda, \mu)$ is 
\begin{align}
\lim_{n\to\infty} \OPT (\lambda, \mu) &=
\min_{t\ge 0}\left\{\sG( 2 \lambda + b \mu t, 4 \rho + b^2\tau) +\gamma^{-1} \sG( b/t, b^2 \gamma \tau)\right\}.
\end{align}
If the minimum above occurs at $t = t_* $ 
such that $\sG'(2\lambda + b\mu t_*, 4\rho + b^2 \tau) = \partial_\kappa \sG(\kappa, 4\rho + b^2\tau)|_{\kappa = 2\lambda + b\mu t_*} > 0$, then
$\lim_{n\to\infty}\OPT(\lambda, \mu) >   \sqrt{4\rho + b^2\tau} + \gamma^{-1}\sqrt{ \frac{\tau}{\gamma}}$. 
 \end{theorem}

\subsection{Proof of Theorem \ref{thm:main_comparison}: the upper bound}
 The following lemma removes
the effect of the projection of $g_x$ ($g_y$) along $v_0$ (resp. $u_0$).
Let $F(x, y) = \frac{1}{n}[\lambda x_1^2 + \<x, g_x\> + b \sqrt{\mu}x_1 y_1 + \<y, g_y\>]$.
Further, let $\widetilde{g}_x$ ($\widetilde{g}_y$) be the vectors obtained
by setting the first coordinate of $g_x$ (resp. $g_y$) to zero, and
$\widetilde{F}(x, y) = \frac{1}{n} [\lambda x_1^2 + \<x, \tg_x\> + b \sqrt{\mu}x_1 y_1 + \<y, \tg_y\>]$.
\begin{lemma}\label{lem:removeprojection}
The optima of $F$ and $\widetilde{F}$ differ by at most $o(1)$. More precisely:
\begin{align}
\Big| \E\max_{x, y} F(x, y) -  \E \max_{x, y}\widetilde{F}(x, y) \Big|  = O\Big( \frac{1}{\sqrt{n}} \Big) \, .\nonumber 
\end{align}
\end{lemma}
\begin{proof}
For any $x, y$:
\begin{align}
&F(x, y) = \frac{1}{n} \Big( \lambda x_1^2 + \<x, g_x\> + \sqrt{\mu} x_1y_1 + \<y, g_y\>\Big) =
\widetilde{F}(x, y) + \frac{1}{n}(x_1 (g_x)_1+ y_1(g_y)_1) \nonumber\\
&\Big| F(x,y) - \widetilde{F}(x,y) \Big| \le \frac{1}{n}( \sqrt{n}|(g_x)_1| + \sqrt{p} |(g_y)_1| ). \nonumber
\end{align}
Maximizing each side over $x, y$ and taking expectation yields the lemma. 
\end{proof}
With this in hand, we can concentrate on
computing the maximum of $\widetilde{F}(x, y)$.
\begin{lemma}\label{lem:maximumFtilde}
Let $\widetilde{g}_x$ ($\widetilde{g}_y$) be the projection of $g_x$ (resp. $g_y$)
orthogonal to the first basis vector. Then
\begin{align}
\lim\sup_{n\do\infty}\E\max_{(x, y)\in \sphere^n\times\sphere^p} \widetilde{F}(x, y) 
& \le \min_{t\le 0}  \sG(2 \lambda + b\mu t, 4 \rho + b^2\tau )
+ \frac{1}{\gamma} \sG(b/t, b^2 \gamma\tau )
\end{align}
\end{lemma}
\begin{proof}
Since $\widetilde{F}(x, y)$ increases if we align the signs of $x_1$ and $y_1$ to $+1$,
we can assume that they are positive. Furthermore, for fixed, positive
$x_1, y_1$, $\widetilde{F}$ is maximized if the other coordinates align with
$\widetilde{g}_x$ and $\widetilde{g}_y$ respectively. Therefore:
\begin{align}
\max_{x, y} \widetilde{F}(x, y) & = \max_{x_1 \in [0, \sqrt{n}], y_1\in [0, \sqrt{p}]}
\frac{\lambda x_1^2}{n} + \sqrt{ 1- \frac{x_1^2}{n}} \frac{\norm{\widetilde{g}_x}}{\sqrt{n}}
 + \frac{b\sqrt{\mu x_1 y_1}}{n} +\sqrt{ 1- \frac{y_1^2}{p}} \frac{\sqrt{p}\norm{\widetilde{g}_y}}{n} \nonumber\\
 & = \max_{m_1, m_2\in [0, 1]}
\lambda m_1 + \sqrt{ 1- m_1} \frac{\norm{\widetilde{g}_x}}{\sqrt{n}}
 + b\sqrt{\frac{\mu m_1 m_2 p}{n}} +\sqrt{ 1- m_2} \frac{\sqrt{p}\norm{\widetilde{g}_y}}{n} \nonumber\\
& \le \max_{m_1, m_2 \in [0, 1]} 
\Big(\lambda + \frac{b \mu t}{2}\Big) m_1 + \sqrt{ 1- m_1} \frac{\norm{\widetilde{g}_x}}{\sqrt{n}} 
 + \frac{p}{n}\Big(\frac{b m_2}{2t} +\sqrt{ 1- m_2} \frac{\norm{\widetilde{g}_y}}{\sqrt{p}}\Big) \nonumber\\
& = \sG(2 \lambda + b \mu t, \norm{\widetilde{g}_x}^2/n)
 +\frac{1}{\gamma} \sG \Big( \frac{b}{t}, \norm{\widetilde{g}_y}^2/p \Big), \label{eq:Ftildeupperbound}
 \end{align} 
 where the first equality is change of variables, the second inequality is the fact that
 $2\sqrt{ab} = \min_{t\ge 0} (at + b/t)$, and the final equality is by direct calculus.

Now let $t_*$ be any minimizer of $\sG(2 \lambda + b\mu t, 4 \rho + b^2 \tau) + \gamma^{-1} \sG(b/t,  b^2 \gamma \tau)$. We may assume that $t_*\not\in\{0, \infty\}$,
otherwise we can use $t_*(\eps)$, an $\eps$-approximate minimizer in $(0, \infty)$ in the argument below.  Since the above holds for any $t$, we have:
\begin{align}
\max_{x, y} \widetilde{F}(x, y) &\le \sG(2 \lambda + b \mu t_*, \norm{\tg_x}^2/n)
+ \gamma^{-1}\sG(b/t_*, \norm{\tg_y}^2/p). 
\end{align}
By the strong law of large numbers, $\norm{\tg_x}^2/n \to 4 \rho + b^2 \tau$ and $\norm{\tg_y}^2/p \to b^2  \gamma \tau$ almost surely. Further, as $\sG(\kappa, \sigma^2)$ is
continuous in the second argument on $(0, \infty)$, when $\kappa \not\in \{0, \infty\}$, almost surely:
\begin{align}
\lim\sup \max_{x, y} \widetilde{F}(x, y) & \le \sG(2\lambda + b \mu t_*, 4 \rho + b^2 \tau)
+ \gamma^{-1} \sG(b/t_*, b^2 \gamma \tau). 
\end{align}
Taking expectations and using bounded convergence yields the lemma.
 \end{proof}

We can now prove the upper bound. 
\begin{proof}[Theorem \ref{thm:main_comparison}, upper bound]
Using Proposition \ref{prop:slepianupper}, Lemma \ref{lem:removeprojection}
and Lemma \ref{lem:maximumFtilde} in order:
\begin{align}
\OPT(\lambda, \mu) & \le \E\{ \max_{x, y} F(x, y) \} \\
& \le \E\{ \max_{x, y} \widetilde{F}(x, y) \} + o(n^{-1/3})\\
& \le \min_{t} \sG(2 \lambda + b\mu t, 4 \rho + b^2 \tau) + \frac{1}{\gamma}\sG( b/t , b^2 \gamma \tau) + o(n^{-1/3}). 
\end{align}
Taking limit $p\to\infty$ yields the result.
\end{proof}
\subsection{Proof of Theorem \ref{thm:main_comparison}: the lower bound}

Recall that $t_*$ denotes  the optimizer of the upper
bound $\sG(2 \lambda + b \mu t, 4 \rho + b^2 \tau) + \gamma^{-1}\sG(b/t, b^2 \gamma \tau)$. By stationarity, we have:
\begin{align}
b\mu \sG'(2 \lambda + b\mu t_*, 4 \rho + b^2 \tau)  - \frac{b}{\gamma t_*^2} \sG'(\frac{b}{t_*}, b^2 \gamma\tau) = 0. \label{eq:stationary_pt}
\end{align}
Now we proceed in two cases. First, suppose $\sG'(2 \lambda + b \mu t_*, 4 \rho + b^2  \tau) = 0$. In this case $\sG'(b/t_*, b^2 \gamma \tau) / t_*^2 = 0$, 
whence $\sG'(b/t_*, b^2 \gamma\tau) = 0$. Indeed, the case when $t_* = \infty$ also satisfies this. However, this also implies that 
$2 \lambda + b \mu t_* \le \sqrt{4 \rho + b^2 \tau}$ and $t_* \ge (\gamma \tau)^{-1/2}$, whereby
$\sG(2 \lambda + b \mu t_*, 4 \rho + b^2 \tau) = \sqrt{ 4 \rho + b^2 \tau}$ and
$\sG'(b /t_*,  b^2 \gamma \tau) = b\sqrt{\gamma \tau}$. In this case
we consider $\tilde{x}, \tilde{y}$ to be the principal
eigenvectors of $W_x, W_y$ rescaled to norms $\sqrt{n}, \sqrt{p}$
respectively and, hence using \eqref{eq:lower},
\begin{align}
\OPT(\lambda, \mu, b) & \ge \frac{1}{2 n} \E\Big[ \< \tilde{x}, W_x \tilde{x} \> + \< \tilde{y}, W_y \tilde{y} \>  \Big] - o(1) .
\end{align}
By standard results on GOE matrices the right hand side
converges to $\sqrt{4 \rho + b^2 \tau} + b \sqrt{\frac{\tau}{\gamma}}$ 
implying the required lower bound. 

Now consider the case that $\sG'(2 \lambda + b \mu t_*, 4 \rho + b^2\tau) > 0$.
Importantly, by stationarity we have
\begin{align}
t_*^2 &= \frac{\sG'(b t_*^{-1}, b^2 \gamma\tau)}{\mu\gamma \sG'(2 \lambda + b \mu t_*, 4\rho + b^2 \tau)}, \label{eq:stationaritytstar}
\end{align}
and that $t_*$ is finite since the numerator is decreasing in $t_*$.
The key ingredient to prove the lower bound
is the following result on the principal eigenvalue/eigenvector
of a deformed GOE matrix. 
\begin{theorem}[\cite{capitaine2009largest,knowles2013isotropic}]\label{thm:rankonedeformed}
Suppose $W\in \reals^{n\times n}$ is a GOE matrix with variance $\sigma^2$, i.e.
$W_{ij} = W_{ji} \sim \normal(0, (1+\delta_{ij}\sigma^2/p)$ and $A = \kappa v_0 v_0^\sT + W$
where $v_0$ is a unit vector.
Then the following holds almost surely and in expectation:
\begin{align}
\lim_{n\to\infty} \lambda_1(A)  = 2\sG(\kappa, \sigma^2) & = \begin{cases}
2\sigma & \text{ if } \kappa < \sigma \\
\kappa +{\sigma^2}/{\kappa} &\text{ if } \kappa > \sigma.
\end{cases}\\
\lim_{n\to\infty} \<v_1(A), v_0\>^2  = 2\sG'(\kappa, \sigma^2) &= \begin{cases}
0 &\text{ if } \kappa <\sigma,\\
1 - {\sigma^2}/{\kappa^2} &\text{ if } \kappa > \sigma.
\end{cases}, 
\end{align}
where $\sG'$ denotes the derivative with respect to the first
argument. 
\end{theorem}
For the prescribed $t_*$, define:
\begin{align}
H(x, y) & = \Big(\lambda + \frac{b\mu t_*}{2} \Big) \frac{\<x, v_0\>^2}{n^2} + \frac{\<x, W_x x\>}{2n}
+ \frac{p}{n}\Big(\frac{b\<y, u_0\>^2}{2 t_* p^2} + \frac{\<y, W_y y\>}{2 p} \Big)
\end{align}
Let $\tilde{x}, \tilde{y}$ be the principal eigenvector of
$(2 \lambda + b \mu t_*) v_0 v_0^\sT /n  + W_x$, $b t_*^{-1} u_0 u_0^\sT/p + W_y$,
rescaled to norm $\sqrt{n}$ and $\sqrt{p}$
respectively. Further, we choose the sign of $\tilde{x}$ 
so that $\<\tilde{x}, v_0\> \ge 0$, and analogously for $\tilde{y}$. 
Now, fixing an $\eps > 0$, we have 
by Theorem \ref{thm:rankonedeformed},   for every $p$ large enough:
\begin{align}
H(\tilde{x}, \tilde{y}) &\ge \sG(2 \lambda + b \mu t_*, 4 \rho + b^2 \tau)
+ \gamma^{-1} \sG(b t_*^{-1}, b^2 \gamma \tau) -\eps \label{eq:Htlowerbound}\\
\frac{\<\tilde{x}, v_0\>}{n} &= \sqrt{2 \sG'(2 \lambda + b \mu t_*, 4 \rho + b^2\tau)} + O(\eps) \label{eq:corrlowerbound} \\
\frac{\<\tilde{y}, u_0\>}{p} & = \sqrt{2 \sG'(b t_*^{-1}, b^2 \gamma \tau)} +O(\eps) \label{eq:corrlowerbound2}
\end{align}
We have, therefore: 
\begin{align}
\OPT(\lambda, \mu,b) & \ge  \E\Big[ H(\tilde{x}, \tilde{y}) + \Big(\frac{b}{n}\sqrt{\frac{\mu}{np} }\<\tilde{x}, v_0\> \<\tilde{y}, u_0\>  - \frac{b\mu t \<\tilde{x}, v_0\>^2  }{2n^2} - \frac{b\<y, u_0\>^2}{2t np} \Big) \Big] \\
&\ge \sG(2 \lambda + b \mu t_*, 4 \rho + b^2 \tau) + \gamma^{-1} \sG(b t_*^{-1}, b^2 \gamma \tau)
+O( \eps (t_*\vee t_*^{-1})) \nonumber\\
&+ \Big( 2\sqrt{ \frac{\mu}{\gamma} \sG'(2 \lambda + b \mu t_*, 4 \rho + b^2 \tau)\sG'(bt_*^{-1}, b^2 \gamma \tau)}
- {b\mu t_* \sG'(2 \lambda + b\mu t_*, 4\rho + b^2\tau)} \nonumber\\
&- \frac{\sG'(bt_*^{-1}, b^2\gamma \tau)}{\gamma t_*}
\Big) \nonumber\\
&\ge\sG(2\lambda + b\mu t_*, 4\rho +b^2 \tau) + \gamma^{-1} \sG(bt_*^{-1},b^2 \gamma \tau)
+O( \eps (t_*\vee t_*^{-1})).
\end{align}
Here the first inequality since we used a specific guess $\tilde{x}, \tilde{y}$, the second using Theorem \ref{thm:rankonedeformed} and 
the final inequality follows since the remainder term vanishes
due to  \myeqref{eq:stationaritytstar}.
Taking expectations and letting $\eps$ going to 0 yields the required
lower bound. 

Given Corollary \ref{cor:zerosignal} and Theorem \ref{thm:main_comparison}, it is not too hard to establish Lemma \ref{lemma:upper_bound}, which we proceed to do next. 
\subsection{Proof of Lemma \ref{lemma:upper_bound}}
\label{sec:proof_upperbdd}
Recall $b_* = \frac{2\mu}{\lambda \gamma} $.  
Part (i) follows directly from Corollary \ref{cor:zerosignal}, upon setting $\rho=\tau=1$, and $b=b_* \sqrt{\gamma}$. To establish part (ii), we use 
Theorem \ref{thm:main_comparison}. In particular, it suffices to establish that with this specific choice of $b = b_*\sqrt{\gamma}$, for any $(\lambda,\mu)$ with $\lambda^2 + \mu^2/\gamma >1$, the minimizer $t_*$ of $G(2\lambda + b\mu t , 4 + b^2) + \gamma^{-1} G(b/t, b^2 \gamma)$ satisfies $G'(2\lambda + b \mu t_*, 4 + b^2 ) >0$.
Let us assume, if possible, that $G(2 \lambda + b \mu t_*, 4 + b^2) =0$. Using the stationary point condition \eqref{eq:stationary_pt}, in this case $G'(b/t_*, b^2 \gamma )=0$. Next, using the definition of $G$ \eqref{def:G}, observe that this implies 
\begin{align}
t_* > \frac{1}{\sqrt{\gamma}},\,\,\,\,  2\lambda + \frac{2 \mu^2}{\lambda \sqrt{\gamma}} t_* < \sqrt{4 + \frac{4\mu^2 }{\lambda^2 \gamma}}. \nonumber
\end{align}
These imply:
\begin{align}
\frac{2}{\lambda}\Big( \lambda^2 + \frac{\mu^2}{\gamma}\Big)
& < 
2\lambda + 2\frac{\mu^2 t_*}{\lambda \mu \sqrt \gamma} \\
&< \sqrt{ 4 + \frac{4\mu^2}{\lambda^2  \gamma}} \\
& = \frac{2}{\lambda}\sqrt{ \lambda^2 + \frac{\mu^2}{\gamma}}.
\end{align}
That this is impossible whenever $\lambda^2 + \frac{\mu^2}{\gamma}>1$. This establishes part (ii). 
To establish part (iii), we again use the upper bound from Proposition \ref{prop:slepianupper}, and note that for $0< \< x, v \> < \tilde{\delta} \sqrt{n}$, 
\begin{align}
\E[\tilde{T}(\tilde{\delta})] \leq \lambda \tilde{\delta}^2 + \sqrt{4 + b_*^2} + \max_{\|y \| =1} \{  b_* \sqrt{\mu} \tilde{\delta} \< u, y  \>  + \frac{1}{\gamma} \< y, g \>\}, \nonumber
\end{align}
where $g \sim \normal(0, b^2 \gamma I_p/p) $. The proof follows using continuity in $\tilde{\delta}$. This completes the proof.

 
\section{Proof of Lemma \ref{lem:densityevolutioninstability}}
Recall the distributional recursion specified by density evolution (Definition \ref{def:densityevolutiondef}).  
\begin{align}
\bar{m}' |_{U} &\stackrel{d}{=} \mu U \E[ V \bar{\eta}] + \zeta_1
 \sqrt{\mu \E[  \bar{\eta} ^2 ]}, \nonumber\\
\bar{\eta}' |_{V'=+1} &\stackrel{d}{=} \frac{\lambda}{\sqrt{d}} \Big[ \sum_{k=1}^{k_+} \bar{\eta}_k |_{+} + \sum_{k=1}^{k_{-}} \bar{\eta}_k|_{-}\Big] - \lambda \sqrt{d} \E[\bar{\eta}]  + \frac{\mu }{\gamma} \E[U \bar{m} ] + \zeta_2 \sqrt{\frac{\mu}{\gamma}\E[\bar{m}^2 ]}, \nonumber 
\end{align}
where $V\sim U(\{\pm 1 \})$, $U\sim \normal (0,1)$, $k_+ \sim \textrm{Poisson}\Big( \frac{d + \lambda \sqrt{d}}{2} \Big)$, $k_{-} \sim \textrm{Poisson} \Big( \frac{d - \lambda \sqrt{d}}{2} \Big)$, $\zeta_1, \zeta_2 \sim \normal(0,1)$ are all mutually independent. Further, $\{\bar{\eta}_k|_{+}\}$  are iid random variables, distributed as $\bar{\eta}|_{V=+1}$. Similarly, $\{\bar{\eta}_k|_{-} \}$, are iid random variables, distributed as $\bar{\eta}|_{V=-1}$. Finally, we require the collections to be mutually independent, and independent of the other auxiliary variables defined above. 

Given these distributional recursions, we compute the vector of moments 
\begin{align}
\E[V'\bar{\eta}'] &= \lambda^2 \E[ V \bar{\eta}] + \frac{\mu}{\gamma } \E[U \bar{m}] \nonumber\\
\E[U'\bar{m}'] &= \mu \E[V \bar{\eta} ] \nonumber\\
\E[\bar{\eta}'^2] &= \lambda^2 \E[\bar{\eta}^2] + \frac{\mu^2}{\gamma^2} \E^2[ U\bar{m}] + \frac{\mu}{\gamma} \E[\bar{m}^2  ] + 2 \frac{\lambda^2}{\gamma} \E[ U \bar{m}] \E[V\etabar]. \nonumber\\
\E[\bar{m}'^2 ] &= \mu^2 \E^2 [V \bar{\eta}] + \mu \E[  \bar{\eta}^2 ] \nonumber
\end{align}
Thus the induced mapping on moments $\phi^{\DE}: \mathbb{R}^4 \to \mathbb{R}^4$, $\phi^{\DE}(z_1, z_2, z_3, z_4) = (\phi_1, \phi_2, \phi_3, \phi_4)$, with 
\begin{align}
\phi_1 &= \lambda^2 z_1 + \frac{\mu}{\gamma} z_2 \nonumber\\
\phi_2 &= \mu z_1 \nonumber\\
\phi_3 &=   \frac{\mu^2}{\gamma^2 } z_2^2  + \frac{2 \lambda^2 }{\gamma} z_1 z_2  + \lambda^2 z_3  + \frac{\mu}{\gamma} z_4  , \nonumber\\
\phi_4 &= \mu^2 z_1^2 + \mu z_3. \nonumber
\end{align}
The Jacobian of $\phi^{\DE}$ at $0$ is, up to
identical row/column permutation:
\begin{align}
J = \left[ \begin{matrix}
\lambda^2 I_2 & \frac{\mu}{\gamma} I_2 \\
\mu I_2 & 0 
\end{matrix} \right]. \nonumber
\end{align}
By direct computation, we see that $z$ is an eigenvalue of $J$ if and only if $z^2 - \lambda^2 z - \frac{\mu^2}{\gamma }=0$. Consider the quadratic function $f(z) = z^2 - \lambda^2 z - \frac{\mu^2}{\gamma}$ and note that $f(0) <0$. Thus to check whether $f$ has a root with magnitude greater than $1$, it suffices to check its value at $z=1, -1$. Note that if $\lambda^2 + \frac{\mu^2}{\gamma }>1$, $f(1) <0$ and thus $J$ has an eigenvalue greater than $1$. Conversely, if $\lambda^2 + \frac{\mu^2}{\gamma } <1$, $f(1) >0$ and $f(-1) = 1+ \lambda^2 - \frac{\mu^2}{\gamma } > 1 - \frac{\mu^2}{\gamma } >0 $. This completes the proof.

\section{Proof of Theorem \ref{thm:graph} }
We prove Theorem \ref{thm:graph} in this Section. 
Recall the matrix mean square errors 
\begin{align}
\MMSE(v; A, B) = \frac{1}{n(n-1)} \E\Big[ \| vv^{T} - \E[vv^{T} | A, B] \|_F^2  \Big], \nonumber\\
\MMSE(v; A^G, B ) = \frac{1}{n(n-1)} \E\Big[ \| vv^{T} - \E[vv^{T}|A^G, B] \|_F^2 \Big]. \nonumber
\end{align}
The following lemma is immediate from Lemma 4.6 in \cite{deshpande2016asymptotic}.
\begin{lemma}\label{lem:overlap_mmse}
Let $\hv = \hv(A, B)$ be any estimator so that
$\norm{\hv}_2 = \sqrt{n}$. Then
\begin{align}
\liminf_{n\to\infty} \frac{\<\hv, v\>}{n} > 0 \text{ in probability }
& \Rightarrow \limsup_{n\to\infty} \MMSE(v; A,B) < 1. 
\end{align}
Furthermore, if $\limsup_{n\to\infty} \MMSE(v; A, B) < 1$, 
there exists an estimator $\hs(A, B)$ with $\norm{\hs(A, B)}_2 = \sqrt{n}$ so that, in probability:
\begin{align}
\liminf_{n\to\infty} \frac{\<\hs, v\>}{n} > 0. 
\end{align}
Indeed, the same holds for the observation model $A^G, B$.
\end{lemma}


\begin{proof}[Proof of Theorem \ref{thm:graph}]
Consider first the case $\lambda^2 + \frac{\mu^2}{\gamma} <1$. 
For any $\theta \in [0, \lambda]$, $\theta^2 + \mu^2/\gamma < 1$ as
well. Suppose we have $A(\theta), B$ according to model
\eqref{eq:gaussiangraphmodel}, \eqref{eq:gaussiancovariatemodel} where
$\lambda$ is replaced with $\theta$. 
 By Theorem \ref{thm:gaussian} (applied at $\theta$) and the second part
of Lemma \ref{lem:overlap_mmse}, $\liminf_{n\to\infty} \MMSE(v; A(\theta), B) =1$. Using the I-MMSE identity \cite{Guo05mutualinformation}, this implies 
\begin{align}
\lim_{n\to\infty} \frac{1}{n} (I(v; A(\theta), B) - I(v; A(0), B) )= \frac{\theta^2}{4} . \label{eq:gaussianinfo}
\end{align}
By Theorem \ref{thm:Universality}, for all $\theta \in [0, \lambda]$
\begin{align}
\lim_{d\to\infty}\lim_{n\to\infty} \frac{1}{n}(I(v; A^G(\theta), B) -
I(v; A^G(0), B) &= \frac{\theta^2}{4}, \label{eq:graphinfo} \\
\text{ and, therefore}
\lim_{n\to\infty} \MMSE(v; A^G, B) & = 1
\end{align}
This implies, via the first part of Lemma 
\ref{lem:overlap_mmse} that for any estimator
$\hv(A^G; B)$, we have $\limsup_{n\to\infty} |\<\hv, v\>|/n = 0$
in probability, as required. 

Conversely, consider the case $\lambda^2 + \frac{\mu^2}{\gamma } >1$. We may assume that $\mu^2/\gamma < 1$, as otherwise
the result follows from Theorem \ref{thm:baik}. Let
$\lambda_0 = (1- \mu^2/\gamma)^{1/2}$.

Now, by the same argument for Eqs.\eqref{eq:gaussianinfo}, \eqref{eq:graphinfo},  we obtain for all $\theta_1, \theta_2 \in [\lambda_0, \lambda]$:
\begin{align}
\limsup_{n\to\infty} \frac{1}{n}( I(v; A(\theta_1), B)
- I(v; A(\theta_2), B) ) & < \frac{\theta_1^2 - \theta_2^2}{4}.
\end{align}
Applying Theorem \ref{thm:Universality}, we have for
all $\theta_1, \theta_2, \theta \in [\lambda_0, \lambda]$:
\begin{align}
\lim_{d\to\infty}\limsup_{n\to\infty} \frac{1}{n}( I(v; A^G(\theta_1), B)
- I(v; A^G(\theta_2), B) ) & < \frac{\theta_1^2 - \theta_2^2}{4}\\
\text{ and therefore, } \limsup \MMSE(v; A^G(\theta), B) < 1. 
\end{align}
Applying then Lemma \ref{lem:overlap_mmse} implies that
we have an estimator $\hs(A^G, B)$ with non-trivial
overlap i.e. in probability:
\begin{align}
\lim_{d\to\infty}\liminf_{n\to \infty} \frac{\<\hs, v\>}{n} > 0 .
\end{align}
This completes the proof. 

\end{proof}

}{}

\iftoggle{arxiv}{
\bibliographystyle{amsalpha}
\bibliography{all-bibliography}
\appendix


\section{Belief propagation: derivation}
\label{sec:bpderivation}

In this section we will derive the belief propagation
algorithm. Recall the observation model for $(A^G, B) \in\reals^{n\times n}\times \reals^{p\times n}$ in Eqs. \eqref{eq:graphmodel}, \eqref{eq:covariatemodel}:
\begin{align}
A^G_{ij} &= \begin{cases}
1 &\text{ with probability } \frac{d + \lambda \sqrt{d}v_i v_j}{n} \\
0 & \text{ otherwise. } 
\end{cases}\\
B_{qi} & = \sqrt{\frac{\mu}{n}} u_q v_i + Z_{qi}, 
\end{align}
where $u_q$ and $Z_{qi}$ are independent $\normal(0, 1/p)$
variables. 

We will use the following conventions throughout this 
section to simplify some of the notation. We will index
nodes in the graph, i.e. elements in $[n]$ with $i, j, k\dots$
and covariates, i.e. elements in $[p]$ with $q, r, s, \dots$.
We will use `$\simeq$' to denote equality of
probability distributions (or densities)
up to an omitted proportionality constant, that may
change from line to line. We will omit
the superscript $G$ in $A^G$. In the graph $G$, 
we will denote neighbors of a node $i$ with $\partial i$
and non-neighbors with $\partial i ^c$. 

We start with the posterior distribution of $u, v$ given
the data $A, B$:
\begin{align}
\d\P\{u , v | A, B\} & = 
\frac{\d\P\{A, B | u , v\}}{\d\P\{A, B\}} \d\P\{u, v\} \\
&\simeq \prod_{i < j} \Big( \frac{d + \lambda \sqrt d v_i v_j}{n}\Big)^{A_{ij}} \Big(1 - \frac{d + \lambda \sqrt d v_i v_j}{n}\Big)^{1- A_{ij}}\nonumber\\
 &\quad\cdot\prod_{q, i} \exp\Big(\sqrt{\frac{\mu p^2}{n}}B_{qi} u_q v_i \Big) \prod_{q} \exp\Big(-\frac{p(1+ \mu)}{2} u_q^2\Big). 
\end{align}
The belief propagation algorithm operates `messages' $\nu^t_{i\to j}, \nu^t_{q\to i},\nu^t_{i\to q}$ which
are probability distributions. 
They represent the marginals of the variables $v_i, u_q$
in the absence of variables $v_j, u_q$, in the posterior
distribtuion $\d\P\{u, v | A, B\}$. We denote by 
$ \E^t_{i\to j}, \E^t_{q\to i}, \E^t_{i\to q}$ expectations with respect
to these distributions. The messages are are
computed using the following update equations:
\begin{align}
\nu^{t+1}_{i\to j}(v_i) &\simeq
\prod_{q\in [p]} \E^t_{q\to i} \Big\{\exp\Big( \sqrt\frac{\mu p^2}{n} B_{qi} v_i u_q\Big) \Big\} %
\prod_{k \in \partial i \bsl j} \E^t_{k\to i} \Big(\frac{d+ \lambda \sqrt{d}v_iv_k}{n}\Big)  \prod_{k \in \partial i ^c \bsl j} \E^t_{k\to i} \Big(1- \frac{d + \lambda \sqrt d v_i v_k}{n}\Big) \,,\label{eq:bpij}\\
\nu^{t+1}_{i\to q}(v_i) &\simeq
\prod_{r\in [p]\bsl q} \E^t_{r\to i} \Big\{\exp\Big( \sqrt\frac{\mu p^2}{n} B_{ri} v_i u_r\Big) \Big\} %
\prod_{k \in \partial i } \E^t_{k\to i} \Big(\frac{d+ \lambda \sqrt{d}v_iv_k}{n}\Big)  \prod_{k \in \partial i ^c } \E^t_{k\to i} \Big(1- \frac{d + \lambda \sqrt d v_i v_k}{n}\Big), \label{eq:bpiq} \\
\nu^{t+1}_{q\to i}(u_q) &\simeq %
\exp\Big(-\frac{p(1+\mu)u_q^2}{2}\Big)
\prod_{j\ne i} \E^t_{j\to q} \Big\{\exp\Big( \sqrt\frac{\mu p^2}{n} B_{qj} v_j u_q\Big) \Big\}\, .\label{eq:bpqi} %
\end{align}
As is standard, we define $\nu^t_i, \nu^t_q$ in the 
same fashion as above, except without the removal 
of the incoming message. 
\subsection{Reduction using Gaussian ansatz}

The update rules \eqref{eq:bpij}, \eqref{eq:bpiq}, 
\eqref{eq:bpqi} are in terms of probability distributions, 
i.e. measures on the real line or $\{\pm 1\}$.  
We reduce them to update rules on real numbers using
the following analytical ansatz. The
measure $\nu^t_{i\to j}$ on $\{\pm 1\}$ 
can be summarized 
using the log-odds ratio: 
\begin{align}
\eta^t_{i\to j} &\equiv \frac{1}{2}\log \frac{\nu^t_{i\to j}(+1)}{\nu^t_{i\to j}(-1)}, \label{eq:etadef}
\end{align}
and we similarly define $\eta^t_{i\to q}$, $\eta^t_i$. In order to 
reduce the densities $\nu^t_{q\to i}$, we use the Gaussian
ansatz:
\begin{align}
\nu^t_{q\to i} &=\normal\Big(\frac{m^t_{q\to i}}{\sqrt{p}}, \frac{\tau^t_{q\to i}}{p}\Big). \label{eq:mdef}
\end{align}
With \Cref{eq:etadef,eq:mdef} we can now
simplify \Cref{eq:bpij,eq:bpiq,eq:bpqi}. The
following lemma computes the inner marginalizations
in \Cref{eq:bpij,eq:bpiq,eq:bpqi}.
We omit the proof. 
\begin{lemma}\label{lem:innermarg}
With $\nu^t, \E^t$ as defined as per \Cref{eq:bpij,eq:bpiq,eq:bpqi}
and $\eta^t, m^t, \tau^t$ as in \Cref{eq:etadef,eq:mdef} we have
\begin{align}
\E^t_{q\to i} \exp\Big(\sqrt\frac{\mu p^2}{n} B_{qi} v_i u_q \Big)
& = \exp\Big( \sqrt\frac{\mu p}{n} B_{q i} v_i m^t_{q\to i} +
\frac{\mu p}{2n} B_{qi}^2 \tau^t_{q\to i}\Big),\label{eq:margqi} \\
\E^t_{i\to j} \Big(\frac{d+\lambda \sqrt{d} v_i v_j}{n}\Big) & = 
\frac{d}{n}\Big(1+ \frac{\lambda v_j}{\sqrt{d}} \tanh(\eta^t_{i\to j})\Big)\, ,\label{eq:margij}\\
\E^t_{i\to j} \Big(1- \frac{d+\lambda \sqrt{d} v_i v_j}{n}\Big) &=  
1 -\frac{d}{n}\Big(1+ \frac{\lambda v_j}{\sqrt{d}} \tanh(\eta^t_{i\to j})\Big)\, ,\label{eq:margij2} \\
\E^t_{i\to q} \exp\Big(p\sqrt\frac{\mu}{n} B_{qi} v_i u_q\Big)
& = \frac{\cosh (\eta^t_{i\to q} + p\sqrt{\mu/n} B_{qi} u_q)}
{\cosh \eta^t_{i\to q}}\label{eq:margiq}.
\end{align}
\end{lemma}
The update equations take a simple form using 
the following definitions
\begin{align}
f(z; \rho) &\equiv \frac{1}{2}\log \Big(\frac{\cosh(z + \rho)}{\cosh(z-\rho)}\Big)\, , \\
\rho &\equiv \tanh^{-1}(\lambda/\sqrt{d})\, , \\
\rho_n & \equiv \tanh^{-1}\Big(\frac{\lambda \sqrt{d}}{n-d}\Big). 
\end{align}
With this, we first compute the update equation for the
node messages $\eta^{t+1}$. Using \Cref{eq:bpij,eq:bpiq,eq:margqi,eq:margij,eq:margij2,eq:margiq}:
\begin{align}
\eta^{t+1}_{i\to j} & = \sqrt\frac{\mu }{\gamma}\sum_{q\in [p]} B_{qi}m^t_{q\to i} + \sum_{k\in \partial i \bsl j } f(\eta^t_{k\to i}; \rho)
- \sum_{k\in \partial i \bsl j} f(\eta^t_{k\to i}; \rho_n) \, ,\\
\eta^{t+1}_{i\to q} & = \sqrt\frac{\mu }{\gamma}\sum_{r\in [p]\bsl q}
B_{ri} m^t_{r\to i} + \sum_{k\in \partial i} f(\eta^t_{k\to i}; \rho) - \sum_{k \in \partial i ^c}f(\eta^t_{k\to i}; \rho_n) \, ,\\
\eta^{t+1}_i & = \sqrt\frac{\mu }{\gamma}\sum_{q\in [p]} B_{qi}m^t_{q\to i} + \sum_{k\in \partial i  } f(\eta^t_{k\to i}; \rho)
- \sum_{k\in \partial i^c } f(\eta^t_{k\to i}; \rho_n)\, .
\end{align}
Now we compute the updates for $m^t_{a\to i}, \tau^t_{a\to i}$. 
We start from \Cref{eq:bpqi,eq:margqi}, and use Taylor
approximation assuming $u_q, B_{jq}$ are both $O(1/\sqrt{p})$, 
as the ansatz \eqref{eq:mdef} suggests. 
\begin{align}
\log \nu^{t+1}_{q\to i}(u_q)
&= {\rm const.} +\frac{-p(1+\mu)}{2} u_q^2
+ \sum_{j\in [n]\bsl i} \log \cosh\Big(\eta^t_{j\to q} + p\sqrt\frac{\mu}{n} B_{qj} u_q \Big) \\
= {\rm const.} + &\frac{-p(1+\mu)}{2} u_q^2 
+ \Big( p\sqrt\frac{\mu}{n}\sum_{j\in[n]\bsl i} B_{qj} \tanh(\eta^t_{j\to q}) \Big) u_q
+ \Big(\frac{p^2 \mu}{2n} \sum_{j\in[n]} B_{qj}^2 %
\sech^2(\eta^t_{j\to q}) \Big) u_q^2 + O\Big(\frac{1}{\sqrt n}\Big).
\end{align}
Note that here we compute $\log \nu^{t+1}$ only up to 
constant factors (with slight abuse of the notation `$\simeq$'). It
follows from this quadratic approximation that:
\begin{align}
\tau^{t+1}_{q\to i} & = \Big(1+\mu - \frac{\mu}{\gamma} \sum_{j\in
[n]\bsl i} B_{qj}^2\sech^2(\eta^t_{j\to q}) \Big)^{-1} \, ,\\
m^{t+1}_{q\to i} & = \tau^{t+1}_{q\to i}\sqrt\frac{\mu}{\gamma}\sum_{j\in[n]\bsl i}B_{qj}\tanh(\eta^t_{j\to q}) \\
& = \frac{  \sqrt{\mu/\gamma}\sum_{j\in[n]\bsl i}B_{qj}\tanh(\eta^t_{j\to q})  }{ 1+\mu - {\mu}{\gamma^{-1} } \sum_{j\in
[n]} B_{qj}^2\sech^2(\eta^t_{j\to q})  }.
\end{align}
Updates computing $m^{t+1}_q, \tau^{t+1}_q$ are analogous.

\subsection{From message passing to approximate message passing}

The updates for $\eta^t, m^t$ derived in the previous section require keeping
track of $O(np)$ messages. In this section, we further 
reduce the number of messages to $O(dn+p)$, i.e. 
linear in the size of the input graph observation. 

The first step is to observe that the dependence
of $\eta^t_{i\to j}$ on $j$ is negligible when $j$
is not a neighbor of $i$ in the graph $G$. 
This derivation is similar to the presentation in
\cite{decelle2011asymptotic}.
As $\sup_{z\in\reals} f(z; \rho) \le \rho$. Therefore, if $i, j$ are not
neighbors in $G$:
\begin{align}
 \eta^{t}_{i\to j} & = \eta^{t}_i - f(\eta^{t-1}_{j\to i}; \rho_n) \\
 & = \eta^{t}_i + O(\rho_n) = \eta^{t}_i+ O\Big(\frac{1}{n}\Big).
 \end{align} 
 Now,  
 for a pair  $i, j$ not connected, by Taylor expansion
 and the fact that  $\partial_z f(z; \rho) \le \tanh(\rho)$, 
 \begin{align}
 f(\eta^t_{i\to j}; \rho_n) - f(\eta^t_i; \rho_n) & =  O\Big(\frac{\tanh (\rho_n )}{n} \Big) = O\Big(\frac{1}{n^2}\Big). 
 \end{align}
Therefore, the update equation for $\eta^{t+1}_{i\to j}$
satisfies:
\begin{align}
\eta^{t+1}_{i\to j} & = \sqrt\frac{\mu }{\gamma}\sum_{q\in[p]} B_{qi}m^t_{q\to i} +
\sum_{k\in \partial i\bsl j} f(\eta^t_{k\to i}; \rho) - \sum_{k\in[n]} f(\eta^t_k; \rho_n) + O\Big(\frac{1}{n}\Big),\label{eq:bpij2} \\
\eta^{t+1}_{i} & = \eta^{t+1}_{i\to j} + f(\eta^t_{j\to i}; \rho).\label{eq:bpi2}
\end{align}
Similarly for $\eta^{t+1}_{i\to q}$ we have:
\begin{align}
\eta^{t+1}_{i\to q} & = \sqrt\frac{\mu }{\gamma}\sum_{r\in [p]\bsl q}
B_{ri} m^t_{r\to i} + \sum_{k\in \partial i} f(\eta^t_{k\to i}; \rho) - \sum_{k \in [n]}f(\eta^t_{k}; \rho_n) + O\Big(\frac{1}{n}\Big). \label{eq:bpiq2}
\end{align}
Ignoring $O(1/n)$ correction term, the update
equations reduce to
 variables $(\eta^t_{i\to j}, \eta^t_i)$ where
$i, j$ are neighbors. 

We now move to reduce updates for $\eta^t_{i\to q}$
and $m^t_{q \to i}$ to involving $O(n)$ variables. This reduction is more subtle then
that of $\eta^t_{i\to j}$, 
where we are able to simply ignore the dependence of $\eta^t_{i\to j}$
on $j$ if $j\not\in\partial i$. We follow a derivation similar
to that in \cite{MontanariChapter}. We use the ansatz:
\def\deta{{\delta\eta}}
\def\dm{{\delta m}}
\def\dtau{{\delta\tau}}
\begin{align}
\eta^t_{i\to q} & = \eta^t_i + \deta^t_{i\to q} \\
m^{t}_{q \to i} & = m^t_q + \dm^t_{q \to i}\\
\tau^t_{q\to i} & = \tau^t_q + \dtau^t_{q\to i},
\end{align}
where the corrections $\deta^t_{i\to q}, \dm^t_{q\to i},
\dtau^t_{q\to i}$
are $O(1/\sqrt{n})$. From \Cref{eq:bpiq2,eq:bpqi} at iteration $t$:
\begin{align}
\eta^{t}_i + \deta^{t}_{i\to q}
&=  \sqrt\frac{\mu }{\gamma} \sum_{r\in [p]\bsl q} B_{ri} (m^{t-1}_r + \dm^{t-1}_{r\to i})
+ \sum_{k\in \partial i} f(\eta^{t-1}_{k\to i}; \rho) - \sum_{k}f(\eta^{t-1}_k; \rho_n) \\
& = \sqrt\frac{\mu }{\gamma} \sum_{r\in [p]} B_{ri} (m^{t-1}_r + \dm^{t-1}_{r\to i})
+ \sum_{k\in \partial i} f(\eta^{t-1}_{k\to i}; \rho) - \sum_{k}f(\eta^{t-1}_k; \rho_n) - \sqrt\frac{\mu}{\gamma}\big(B_{qi}m^{t-1}_{q} + B_{qi} \dm^{t-1}_{q\to i}\big).
\end{align}
Notice that the last term is the only term that depends
on $q$. Further, since $B_{qi}\dm^{t-1}_{q\to i} = O(1/n)$ by 
our ansatz, we may safely ignore it to obtain
\begin{align}
\eta^{t}_i & = \sqrt\frac{\mu }{\gamma} \sum_{r\in [p]} B_{ri} (m^{t-1}_r + \dm^{t-1}_{r\to i})
+ \sum_{k\in \partial i} f(\eta^{t-1}_{k\to i}; \rho) - \sum_{k}f(\eta^{t-1}_k; \rho_n) \label{eq:etaiden} \\
\deta^{t}_{i\to q} & = -\sqrt\frac{\mu}{\gamma} B_{qi}m^{t-1}_q.
\label{eq:detaiden}
\end{align}
We now use the update equation for $\tau^{t+1}_{q\to i}$:
\begin{align}
\tau^{t+1}_{q} &= \left(1+\mu - \frac{\mu}{\gamma} \sum_{j\in [n]} B_{qj}^2\sech^2(\eta^t_j + \deta^t_{j\to q}) \right)^{-1}+O(1/n)\\
& = \left(1 +\mu - \frac{\mu}{\gamma} \sum_{j\in [n]} B_{qj}^2\big((\sech^2(\eta^t_j) - 2 \sech^2(\eta^t_j)\tanh(\eta^{t}_j) \deta^t_{i\to q}\big)\right)^{-1}+O(1/n), 
\end{align}
where we expanded the equation to linear order in $\deta^t_{i\to q}$ and ignored higher order terms. By the identification 
\Cref{eq:detaiden}:
\begin{align}
\tau^{t+1}_{q} 
& = \left(1 +\mu - \frac{\mu}{\gamma} \sum_{j\in [n]} B_{qj}^2\sech^2(\eta^t_j) + 2\Big(\frac{\mu}{\gamma}\Big)^{3/2}
\sum_{j\in [n]} B_{qj}^3\sech^2(\eta^t_j)\tanh(\eta^{t}_j) m^{t-1}_q \right)^{-1}+O(1/n). 
\end{align}
Notice here, that there is no term that explicitly depends
on $i$ and the final term is $O(1/\sqrt{n})$ since $B_{qj} =O(1/\sqrt{n})$. Therefore, ignoring lower order terms, we have the identification:
\begin{align}
\tau^{t+1}_q & = \left(1 + \mu - \frac{\mu}{\gamma} \sum_{j\in [n]} B_{qj}^2 \sech^2(\eta^t_j)\right)^{-1},\label{eq:tauiden}\\
\dtau^{t+1}_{q\to i} & = 0 \label{eq:dtauiden}.
\end{align}
Now we simplify the update for $m^{t+1}_{q\to i}$
using Taylor expansion to first order:
\begin{align}
m^{t+1}_q + \dm^{t+1}_{q\to i} & = 
 \frac{\sqrt{\mu/\gamma}}{\tau^{t+1}_q}\sum_{j\in[n]\bsl i}B_{qj}\tanh(\eta^t_{j} + \deta^t_{j\to q})   \\
 & =   \frac{\sqrt{\mu/\gamma}}{\tau^{t+1}_q}\sum_{j\in[n]\bsl i}\left(B_{qj}\tanh(\eta^t_{j})  + B_{qj}\sech^2(\eta^t_i)\deta^t_{j\to q} \right)\\
 & = \frac{\sqrt{\mu/\gamma}}{\tau^{t+1}_q}   \sum_{j\in [n]\bsl i} \Big(B_{qj} \tanh(\eta^t_j) - \sqrt\frac{\mu}{\gamma} B_{qj}^2 \sech^2(\eta^t_j) m^{t-1}_q\Big) \\
 & =  \frac{\sqrt{\mu/\gamma}}{\tau^{t+1}_q} \sum_{j\in [n]} B_{qj}\tanh(\eta^t_j)- \frac{\mu}{\gamma\tau^{t+1}_q} %
  \bigg(\sum_{j\in [n]} B_{qj}^2 \sech^2(\eta^t_j) \bigg) m^{t-1}_q  \nonumber \\
&\quad - \frac{\sqrt{\mu/\gamma}}{\tau^{t+1}_q}   \big(B_{qi} \tanh(\eta^t_i) - \sqrt{\mu/\gamma} B_{qi}^2 \sech^2(\eta^t_i) m^{t-1}_q \big).
\end{align}
Only the final term is dependent on $i$, therefore we can
identify:
\begin{align}
m^{t+1}_q & = \frac{\sqrt{\mu/\gamma}}{\tau^{t+1}_q}
 \sum_{j\in [n]} B_{qj}\tanh(\eta^t_j)- {\frac{\mu}{\gamma\tau^{t+1}_q}} \bigg(\sum_{j\in [n]} B_{qj}^2 \sech^2(\eta^t_j) \bigg) m^{t-1}_q, \label{eq:miden}\\
 \dm^{t+1}_{q\to i} & = -\frac{\sqrt{\mu/\gamma}}{\tau^{t+1}_q} B_{qi}\tanh(\eta^t_i ).\label{eq:dmiden}
\end{align}
Here, as before, we ignore the lower order term in $\dm^{t+1}_{q\to i}$. Now we can substitute the identification \Cref{eq:dmiden}
back in \Cref{eq:etaiden} at iteration $t+1$:
\begin{align}
\eta^{t+1}_i & = %
\sqrt\frac{\mu }{\gamma} \sum_{r\in [p]} B_{ri} m^{t}_r %
-\frac{\mu}{\gamma}  \sum_{r\in [p]} \frac{B_{ri}^2}{\tau^t_r} \tanh(\eta^{t-1}_i)%
+ \sum_{k\in \partial i} f(\eta^{t}_{k\to i}; \rho) %
- \sum_{k}f(\eta^{t}_k; \rho_n).
\end{align}
Collecting the updates for $\eta^t_i, \eta^t_{i\to j}, m^t_q$
we obtain the approximate message passing algorithm:
\begin{align}
\eta^{t+1}_{i} & = \sqrt{\frac{\mu}{\gamma}} \sum_{q\in [p]}
B_{qi}m^t_q -%
 \frac{\mu}{\gamma}\bigg( \sum_{q\in [p]}\frac{ B_{qi}^2}{\tau^t_q}\bigg) \tanh(\eta^{t-1}_i) + %
\sum_{k\in \partial i} f(\eta^t_{k\to i}; \rho) - \sum_{k\in[n]} f(\eta^t_k; \rho_n) \, ,\label{eq:amp1}\\
\eta^{t+1}_{i\to j} & =  \sqrt{\frac{\mu}{\gamma}} \sum_{q\in [p]}
B_{qi}m^t_q -%
 \frac{\mu}{\gamma}\bigg( \sum_{q\in [p]}\frac{ B_{qi}^2}{\tau^t_q}\bigg) \tanh(\eta^{t-1}_i) + %
\sum_{k\in \partial i\setminus j} f(\eta^t_{k\to i}; \rho) - \sum_{k\in[n]} f(\eta^t_k; \rho_n)\, ,\label{eq:amp2} \\
m^{t+1}_q & = \frac{\sqrt{\mu/\gamma}}{\tau^{t+1}_q}%
 \sum_{j\in [n]} B_{qj}\tanh(\eta^t_j)- %
 {\frac{\mu}{\gamma\tau^{t+1}_q}} %
 \bigg(\sum_{j\in [n]} B_{qj}^2 \sech^2(\eta^t_j) \bigg)%
  m^{t-1}_q \label{eq:amp3} \\
\tau^{t+1}_q & = \left(1 + \mu - \frac{\mu}{\gamma} \sum_{j\in [n]} B_{qj}^2 \sech^2(\eta^t_j)\right)^{-1}\label{eq:amp4}.
\end{align}
\subsection{Linearized approximate message passing}
This algorithm results from expanding the updates 
\Cref{eq:amp1,eq:amp2,eq:amp3,eq:amp4} to linear
order in the messages $\eta^t_i, \eta^t_{i\to j}$:
\begin{align}
\eta^{t+1}_{i} & = \sqrt{\frac{\mu}{\gamma}} \sum_{q\in [p]}
B_{qi}m^t_q -%
 \frac{\mu}{\gamma}\bigg( \sum_{q\in [p]}\frac{ B_{qi}^2}{\tau^t_q}\bigg) \eta^{t-1}_i + %
\frac{\lambda}{\sqrt{d}}\sum_{k\in \partial i} \eta^t_{k\to i} %
- \frac{\lambda\sqrt{d}}{n}\sum_{k\in[n]} \eta^t_k  \label{eq:linearamp1}\\
\eta^{t+1}_{i\to j} & = \sqrt{\frac{\mu}{\gamma}} \sum_{q\in [p]}
B_{qi}m^t_q -%
 \frac{\mu}{\gamma}\bigg( \sum_{q\in [p]}\frac{ B_{qi}^2}{\tau^t_q}\bigg) \eta^{t-1}_i + %
\frac{\lambda}{\sqrt{d}}\sum_{k\in \partial i\setminus j} \eta^t_{k\to i} %
- \frac{\lambda\sqrt{d}}{n}\sum_{k\in[n]} \eta^t_k 
\label{eq:linearamp2} \\
m^{t+1}_q & = \frac{\sqrt{\mu/\gamma}}{\tau^{t+1}_q}%
 \sum_{j\in [n]} B_{qj}\eta^t_j- %
 \frac{\mu}{\gamma \tau^{t+1}_q} %
 \bigg(\sum_{j\in [n]} B_{qj}^2  \bigg)%
  m^{t-1}_q \label{eq:linearamp3} \\
\tau^{t+1}_q & = \left(1 + \mu - \frac{\mu}{\gamma} \sum_{j\in [n]} B_{qj}^2 \right)^{-1}\label{eq:linearamp4}.
\end{align}
This follows from the linear approximation  $f(z; \rho)  = \tanh(\rho) z$ for small $z$. The algorithm given in the main text follows
by using the law of large numbers to approximate $\sum_{j\in [n]}B_{qj}^2 \approx 1/\gamma$, $\sum_{q\in [p]}B_{qj}^2 \approx 1$, and hence $\tau^^t_q\approx 1$.

}

\end{document}